\documentclass[12pt,notcite, notref]{article}

\usepackage{graphics,amsmath,amssymb,amsthm,accents}
\usepackage{epic}

\newtheorem{lem}{Lemma}

\newtheorem{defn}{Definition}
\def \h#1{\widehat{#1}}
\def \t#1{\widetilde{#1}}
\def \b#1{\overline{#1}}
\def \c#1{\accentset{\circ}{#1}}
\def \c#1{\accentset{\circ}{#1}}
\def \th#1{\widehat{\widetilde{#1}}}
\def \hb#1{{\widehat{\overline{#1}}}}
\def \bh#1{{\widehat{\overline{#1}}}}
\def \bb#1{{\overline{\overline{#1}}}}
\def \tb#1{\widetilde{\overline{#1}}}
\def \bt#1{\widetilde{\overline{#1}}}
\def \dh#1{\underaccent{\hat}{#1}}
\def \db#1{\underaccent{\bar}{#1}}
\def \dt#1{\underaccent{\tilde}{#1}}
\def \dth#1{\underaccent{\tilde}{\underaccent{\hat}{#1}}}
\def \dtb#1{\underaccent{\bar}{\underaccent{\tilde}{#1}}}
\def \dhb#1{\underaccent{\bar}{\underaccent{\hat}{#1}}}
\def \dbb#1{\underaccent{\bar}{\underaccent{\bar}{#1}}}
\numberwithin{equation}{section}
\newtheorem{prop}{Proposition}

\title{Soliton Solutions for ABS Lattice Equations: \\
II: Casoratians and Bilinearization}

\author{Jarmo Hietarinta$^1$\footnote{E-mail: jarmo.hietarinta@utu.fi}
  ~ and ~ Da-jun Zhang$^2$\footnote{E-mail: djzhang@staff.shu.edu.cn }
  \\
  {\small\it $^1$Department of Physics and Astronomy,
    University of Turku, FIN-20014 Turku, Finland} \\
  {\small\it $^2$Department of Mathematics, Shanghai University,
    Shanghai 200444, P.R. China}}

\date{\today}

\begin{document}

\maketitle

\begin{abstract}
  In Part I soliton solutions to the ABS list of multi-dimensionally
  consistent difference equations (except Q4) were derived using
  connection between the Q3 equation and the NQC equations, and then by
  reductions. In that work central role was played by a Cauchy matrix.
  In this work we use a different approach, we derive the $N$-soliton
  solutions following Hirota's direct and constructive method. This
  leads to Casoratians and bilinear difference equations. We give here
  details for the H-series of equations and for Q1; the results for
  Q3 have been given earlier.
\end{abstract}

\section{Introduction}
The analysis of integrability for discrete systems is now in active
development. Discreteness introduces many complications in comparison
with continuous integrability, mainly due to the lack of Leibniz rule
for discrete derivatives. This also includes the fact that space-time
itself has many discretizations.

The basic philosophy and definitions have been given in Part I
\cite{NAJ}, here we repeat only some essential ingredients.  The
underlying space is formed by a Cartesian square lattice and the
dynamical equation is defined on an elementary square of this lattice.
(There are other possible settings but even this case has not been
fully analyzed.) As for the definition of integrability we choose
``Consistency-around-the cube'' (CAC) which is further explained in
Section \ref{S:CAC}. With this choice of integrability many nice
properties follow, including the existence of a Lax pair.
Furthermore, with mild additional assumption one can classify the
integrable models \cite{ABS-CMP-2002} and the ``ABS list'' is
surprisingly short.

In this paper we construct multi-soliton solutions for the H-series of
models in the ABS-list, as well as for the Q1 model.  Our approach is
constructive for each model and is based on the approach of Hirota, as
explained in Section \ref{S:H}. Some work on this problem already
exists for the two highest members of the ABS-list, but the Q4 results
in \cite{Sol-Q4-2007} is far from explicit and the result for Q3
\cite{Sol-Q3-2008} relied on a specific association with the
NQC-model, which has been further elaborated in Part I.

We hope that the detailed analysis of the simpler models, for which
everything can be made systematic and explicit, will provide further
understanding on the soliton question and that it could be used for
other models as well.

In the next section we will give the general background for our approach
and then in subsequent sections we will go through the models H1, H2,
H3 and Q1.

\section{Generalities}
\subsection{Multidimensional consistency and the ABS  list\label{S:CAC}}
As in Part I  we only consider quadrilateral lattice equations
defined by a multi-linear relation on the values at the four corners
of an elementary square of Cartesian lattice, see Figure \ref{F:1}(a).
\begin{figure}[b]
\setlength{\unitlength}{0.0004in}
\hspace{2cm}
\begin{picture}(3482,2813)(0,-10)
\put(-800,1510){\makebox(0,0)[lb]{$(a)$}}
\put(1275,2708){\circle*{150}}
\put(825,2808){\makebox(0,0)[lb]{$\h u$}}
\put(3075,2708){\circle*{150}}
\put(3375,2808){\makebox(0,0)[lb]{$\h{\t u}$}}
\put(1275,908){\circle*{150}}
\put(825,1008){\makebox(0,0)[lb]{$u$}}
\put(3075,908){\circle*{150}}
\put(3300,1008){\makebox(0,0)[lb]{$\t u$}}
\drawline(275,2708)(4075,2708)
\drawline(3075,3633)(3075,0)
\drawline(275,908)(4075,908)
\drawline(1275,3633)(1275,0)
\end{picture}
\hspace{4cm}
\begin{picture}(3482,3700)(0,-500)
\put(-1200,1000){\makebox(0,0)[lb]{$(b)$}}
\put(450,1883){\circle{150}}
\put(-100,1883){\makebox(0,0)[lb]{$\b{\t u}$}}
\put(1275,2708){\circle*{150}}
\put(825,2708){\makebox(0,0)[lb]{$\b u$}}
\put(3075,2708){\circle{150}}
\put(3375,2633){\makebox(0,0)[lb]{$\b{\h u}$}}
\put(2250,83){\circle{150}}
\put(2650,8){\makebox(0,0)[lb]{$\h{\t u}$}}
\put(1275,908){\circle{150}}
\put(1275,908){\circle*{90}}
\put(825,908){\makebox(0,0)[lb]{$u$}}
\put(2250,1883){\circle{150}}
\put(2250,1883){\circle{220}}
\put(2250,1883){\circle{80}}
\put(1850,2108){\makebox(0,0)[lb]{$\b{\h{\t u}}$}}
\put(450,83){\circle*{150}}
\put(0,8){\makebox(0,0)[lb]{$\t u$}}
\put(3075,908){\circle*{150}}
\put(3300,833){\makebox(0,0)[lb]{$\h u$}}
\drawline(1275,2708)(3075,2708)
\drawline(1275,2708)(450,1883)
\drawline(450,1883)(450,83)
\drawline(3075,2708)(2250,1883)
\drawline(450,1883)(2250,1883)
\drawline(3075,2633)(3075,908)
\dashline{60.000}(1275,908)(450,83)
\dashline{60.000}(1275,908)(3075,908)
\drawline(2250,1883)(2250,83)
\drawline(450,83)(2250,83)
\drawline(3075,908)(2250,83)
\dashline{60.000}(1275,2633)(1275,908)
\end{picture}\label{F:1}
\caption{(a): The points on which the map is defined, and (b): the
  consistency cube.}
\end{figure}
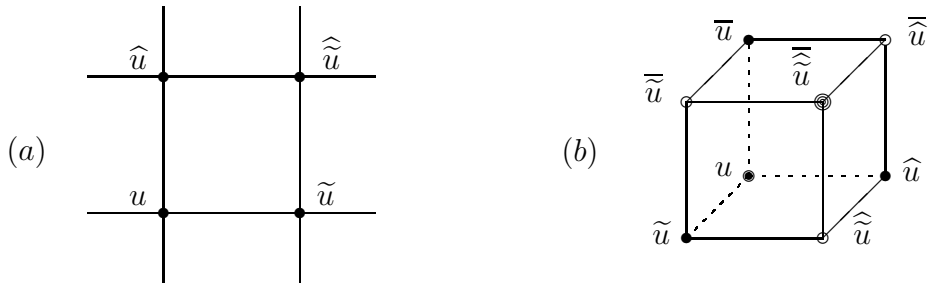
Given a base point $u_{nm}=u$ we indicate the shifts in the $n,m$
directions by a tilde or a hat as described in Figure \ref{F:1}.

Multidimensional consistency is an essential ingredient in the
construction of soliton solutions, the new dimensions standing for
parameters of the solitons. The consistency is defined as follows: We
adjoin a third direction and construct an elementary cube as in Figure
\ref{F:1}(b), the shifts in the new third direction are denoted by a
bar, e.g., $u_{n,m,k+1}=\b u$.  On the base square we have the equation
$Q(u,\t u,\h u,\h{\t u};p,q)=0$, and we now introduce the same
equation with different variables on the sides and on the top, thus we
will have altogether the following set of equations
\begin{subequations}
\label{CACeq}
\begin{eqnarray}
  Q(u,\t u,\h u,\h{\t u};p,q)=0, &&
Q(\b u,\b{\t u},\b{\h u},\b{\h{\t u}};p,q)=0,\label{CACeq-th}\\
  Q(u,\t u,\b u,\b{\t u};p,r)=0, &&
 Q(\h u,\h{\t u},\b{\h u},\b{\h{\t u}};p,r)=0,\label{CACeq-tb}\\
  Q(u,\h u,\b u,\b{\h u};q,r)=0, &&
 Q(\t u,\h{\t u},\b{\t u},\b{\h{\t u}};q,r)=0.\label{CACeq-bh}
\end{eqnarray}
\end{subequations}
Now considering Figure \ref{F:1}(b) we can have initial values given
at black circles ($u,\,\t u,\,\h u,\,\b u$), and in order to compute
the remaining four values we have six equations. We can use the LHS
equations in \eqref{CACeq-th}, \eqref{CACeq-tb}, \eqref{CACeq-bh} to
compute $\h{\t u},\,\b{\t u},\, \b{\h u}$, respectively, and this
leaves the three RHS equations from which we should get the same value for
$\b{\h{\t u}}$, and this implies two consistency conditions.

In \cite{ABS-CMP-2002} a classification of quadrilateral lattice
equations was performed following the above definition of
integrability, with two additional requirements on the equations:
symmetry and the so called ``tetrahedron property'', which means that
the computed value for $\b{\h{\t u}}$ should only depend on $\t u,\,\h
u,\,\b u$ and not on $u$. The ABS list is the following:

\[(u-\h{\t u}) (\t u-\h u)+q-p =0,\eqno{\text{(H1)}} \]

\[(u-\h{\t u}) (\t u-\h u)+
(q-p ) (u+\t u+\h u+\h{\t u})+q^2-p ^2 =0,\eqno{\text{(H2)}} \]

\[p  (u \t u+\h u \h{\t u})-
q (u \h u+\t u \h{\t u})+\delta (p ^2-q^2) =0,
\eqno{\text{(H3)}} \]

\[p (u+\h u) (\t u+\h{\t u}) - q(u+\t u)
(\h u+\h{\t u})- \delta^2 p  q (p -q) =0,
\eqno{\text{(A1)}} \]

\[(q^2-p^2) (u\t u \h u \h{\t u}+1)+
q (p^2-1) (u \h u+\t u \h{\t u})-
p (q^2-1) (u \t u+\h u \h{\t u}),
\eqno{\text{(A2)}} \]

\[p  (u-\h u) (\t u-\h{\t u})-
q (u-\t u) (\h u-\h{\t u})+\delta^2 p  q (p-q) =0,
\eqno{\text{(Q1)}} \]

\[\begin{array}{l}p  (u-\h u) (\t u-\h{\t u})-
q (u-\t u) (\h u-\h{\t u})\\ \quad +
p  q (p -q) (u+\t u+\h u+\h{\t u})
-p  q (p -q) (p ^2-p  q+q^2) =0,\end{array}
\eqno{\text{(Q2)}}\]

\[\begin{array}{l}(q^2-p ^2) (u \h{\t u}+\t u \h u)
+q (p ^2-1) (u \t u+\h u \h{\t u})
-p  (q^2-1) (u \h u +\t u \h{\t u})\\ \quad
-\delta^2 (p ^2-q^2) (p ^2-1) (q^2-1)/(4 p  q) =0,\end{array}
\eqno{\text{(Q3)}}\]

\[\begin{array}{l}
  (h(p)f(q) - h(q)f(p)) [(u \t u  \h u \h{\t u} + 1) f(p)
  f(q)- (u \h{\t u} + \t u \h u)]\\ \quad
  + (f(p)^2 f(q)^2 -1) [(u \t u + \h u \h{\t u} )f(p) - (u \h u + \t u
  \h{\t u}) f(q)]=0,\end{array}
\eqno{\text{(Q4)}}\]
where $h^2=f^4+\delta f^2+1$. This can be parameterized with Jacobi
elliptic functions: $f(x)=k^\frac14\,{\rm sn}(x,k),\,h(x)={\rm
  sn}'(x,k),\, \delta=-(k+1/k)$.
The simpler form of Q4 given above was actually discovered later in
\cite{JH-JNMP-cac}.

The A-series is auxiliary in the sense that A1 goes to Q1 with $u\to
(-1)^{n+m} u$ and A2 to $Q3_{\delta=0}$ by $u\to u^{(-1)^{n+m}}$.
We will not discuss them here.

\subsection{Hirota's bilinear method\label{S:H}}
Our approach to soliton solutions is based on Hirota's bilinear
method\cite{Hirota-book}. It has been very successful in the continuous
case and it is expected to be equally fruitful in constructing
discrete multi-soliton solutions. Hirota's idea was to make a dependent
variable transform into new variables, for which the soliton solution
would be given by a polynomial of exponentials. In terms of these new
dependent variables the dynamical equations were quadratic and
derivatives appeared only in terms of Hirota's bilinear derivatives.

But what would be the natural discrete generalization? The essential
property of an equation in Hirota's bilinear form seems to be its gauge
invariance (cf., \cite{tril} in the continuous case) and it has a
natural discrete extension, leading us to the following definition:
\begin{defn} We say  an equation is in {\bf Hirota bilinear (HB) form}
  if it can be written as
  \begin{equation}
    \label{eq:HB}
    \sum_j\, c_j\, f_{j}(n+\nu_{j}^+,m+\mu_{j}^+)\,
g_{j}(n+\nu_{j}^-,m+\mu_{j}^-)=0
  \end{equation}
  where the index sums $\nu_{j}^++\nu_{j}^-=\nu^s ,
  \mu_{j}^++\mu_{j}^-=\mu^s$ do not depend on $j$.  (The functions
  $f,g$ may be the same.)
\end{defn}

\begin{prop}
\label{P:gauge} Equations in HB form are gauge invariant, i.e.,
  if functions $f_j,g_j$ solve a set of equations in HB form, then so
  do the {\em gauge transformed} functions
\begin{equation}  \label{eq:gauge}
  f'_j(n,m)=A^nB^m\, f_j(n,m),\quad  g'_j(n,m)=A^nB^m\, g_j(n,m).
\end{equation}
\end{prop}
\begin{proof}
  We find
\[
f'_{j}(n+\nu_{j}^+,m+\mu_{j}^+)\, g'_{j}(n+\nu_{j}^-,m+\mu_{j}^-)
=A^{2n+\nu_{j}^++\nu_{j}^-}B^{2m+\mu_{j}^++\mu_{j}^-}
f_{j}(n+\nu_{j}^+,m+\mu_{j}^+)\, g_{j}(n+\nu_{j}^-,m+\mu_{j}^-)
\]
but since the overall factor is the same in each term of the $j$-sum
in \eqref{eq:HB} it can be taken out.
\end{proof}

For integrable equations in HB form there is a perturbative technique
which leads to multi-soliton solutions, more or less algorithmically.
This is described in the next section.

\subsection{Constructing background solutions}
First we have to construct the background or vacuum or seed solution,
on top of which the soliton solutions are constructed. To do this
we use the fixed-point idea \cite{Sol-Q4-2007} in which the consistent
equations \eqref{CACeq} are used with the assumption that $u=\b u$.
However, since some equations are invariant under $u\to T(u)$ it is
actually sufficient that $\b u=T(u)$.  We only consider global
invariances (that is, the transformation $T$ is independent of $n,m$)
and in order to keep the multi-linearity we assume that the
transformation is linear fractional:
\begin{equation}
  \label{eq:moe}
  u_{nm}\to T(u_{nm}):=\frac{c_1u_{nm}+c_2}{c_3u_{nm}+c_4},
\quad c_1c_4-c_2c_3\neq 0.
\end{equation}
Depending on the equation there will be conditions on the parameters
$c_i$. It is straightforward to find the invariances of the equations,
they are given in Table \ref{T:inv}. Note that the special cases with
$\delta=0$ have a bigger invariance group. For Q4 there are various
special cases depending on which limiting case of the elliptic
parameterization one chooses.
\begin{table}
  \centering
  \begin{tabular}{|c|c|c|}
\hline
    Eq. & $T(u)$ & $T(u)$ when $\delta=0$ \\
\hline    H1 &  $u+c,\, -u+c$ & NA\\
\hline    H2 & $u$ & NA\\
\hline    H3 & $u,\,-u$ & $c u,c/u$\\
\hline    A1 & $u,\,-u$  & $c u,c/u$\\
\hline    A2 & $u,\,-u,\,1/u,\,-1/u$ & NA \\
\hline    Q1 &  $u+c,\, -u+c$ & full M\"obius\\
\hline    Q2 & $u$ & NA\\
\hline    Q3 & $u,\,-u$  & $c u,c/u$ \\
\hline    Q4 & $u,\,-u,\,1/u,\,-1/u$ & various \\
\hline
  \end{tabular}
  \caption{Invariances \eqref{eq:moe} of the equations in the ABS list.}
  \label{T:inv}
\end{table}

 For each invariance  $T$ of a given equation  we then need to solve
\begin{subequations}
  \label{CAC-bg}
\begin{eqnarray}
    Q(u,\t u,T(u),T({\t u});p,r)&=&0,\label{CACbg-t}\\
  Q(u,\h u,T( u),T({\h u});q,r)&=&0,\label{CACbg-h}
\end{eqnarray}
\end{subequations}
in order to obtain the corresponding background solution (here $r$ is
a parameter of the solution). Such a solution then automatically
solves \eqref{CACeq-th} by virtue of CAC.

It may be possible to construct still further solutions that can be
called ``background solutions'', but we will not consider them here.

\subsection{Constructing 1-soliton solutions\label{S:1ss}}
Once the background solution is obtained we construct a one-soliton
solution (1SS) using the CAC cube once more, with $\b u$ now being the
1SS. This amounts, in fact, to a B\"acklund transformation (BT).

Thus the first task is to solve
\begin{subequations}
  \label{CAC-1ss}
\begin{eqnarray}
    Q(u,\t u,\b u,\b{\t u};p,r)&=&0,\label{CAC1ss-t}\\
  Q(u,\h u,\b u,\b{\h u};q,r)&=&0,\label{CAC1ss-h}
\end{eqnarray}
\end{subequations}
where we take
\begin{equation}
\b{u}=\b{u}_{0}+v.
\label{u-bar-gen}
\end{equation}
Here $\b u_0$ is the background solution and the bar-shift stands for
possible modifications to get the $v$-equations into a suitable form.
Then solving for $\t v$ and $\h v$ from \eqref{CAC-1ss} we get (for
quadratic $Q$)
\begin{equation}
\t{v}=\frac{Ev}{v+F},\quad \h{v}=\frac{Gv}{v+H},
\label{v-gen-1}
\end{equation}
where $E,F,G,H$ may depend on $n,m$. The bar-shift modification in
$u_0$ is to be chosen so that the numerators have no constant term.
By introducing $v=g/f$ and $\Phi=(g,f)^T$ we can write \eqref{v-gen-1}
as a matrix equation
\begin{equation}
\Phi(n+1,m)=\mathcal N(n,m)\Phi(n,m),\quad
\Phi(n,m+1)=\mathcal M(n,m)\Phi(n,m),
\end{equation}
where
\begin{equation}\label{NM-matricesEG}
\mathcal N(n,m)=\Lambda\begin{pmatrix}E&0\\1&F\end{pmatrix},\quad
\mathcal M(n,m)=\Lambda'\begin{pmatrix}G&0\\1&H\end{pmatrix},
\end{equation}
where $E,F,G,H$ may depend on $n,m$.  The constants of separation
$\Lambda,\Lambda'$ are to be chosen so that
$\widehat{\widetilde{\Phi}} =\widetilde{\widehat{\Phi}}$. In all cases
studied in this paper it turns out that
\begin{equation}\label{NM-matricesST}
\mathcal N(n,m)=\begin{pmatrix}S\frac{U_{n+1,m}}{U_{n,m}}&0\\
\frac{\sigma}{U_{n,m}}&\Delta\end{pmatrix},\quad
\mathcal M(n,m)=\begin{pmatrix}T\frac{U_{n,m+1}}{U_{n,m}}&0\\
\frac{\tau}{U_{n,m}}&\Omega\end{pmatrix},
\end{equation}
where $S,\,T,\,\sigma,\,\tau,\,\Delta,\,\Omega$ are constants and
$U_{nm}$ some function of $n,m$ (sometimes $U_{nm}=1$).  If we now
introduce $\Psi$ by
\begin{equation}
  \label{psi-def}
\Phi(n,m)=\begin{pmatrix}U_{n,m}&0\\0&1\end{pmatrix}\Psi(n,m),
\end{equation}
then for $\Psi$ we have
\begin{equation}
  \label{psi-eqs}
\Psi(n+1,m)=\begin{pmatrix}S&0\\ \sigma &\Delta\end{pmatrix}\Psi(n,m),\quad
\Psi(n,m+1)=\begin{pmatrix}T&0\\ \tau &\Omega\end{pmatrix}\Psi(n,m).
\end{equation}
These are compatible if
\begin{equation}
  \label{NM-compat}
\frac{\sigma}{S-\Delta}=\frac{\tau}{T-\Omega},
\end{equation}
and in that case it is easy to derive
\begin{equation}
  \label{psi-res}
\Psi(n,m)=\begin{pmatrix} S^nT^m & 0 \\
\frac{\tau}{T - \Omega}(S^n T^m-\Delta^n \Omega^m) &
 \Delta^n\Omega^m \end{pmatrix}\Psi(0,0).
\end{equation}
{}From this we can construct $\Phi(n,m)$ and then $v$: Introduce
\begin{equation}
  \label{rho-def}
\rho_{n,m}=\biggl(\frac{S}{\Delta}\biggr)^n
\biggl(\frac{T}{\Omega}\biggr)^m \rho_{0,0}
\end{equation}
and then $v$ can be written as
\begin{equation}
v_{n,m}=\frac{U_{n,m}\frac{v_{0,0}}{U_{0,0}}\rho_{n,m}/\rho_{0,0} }
{1+\frac{\tau}{T-\Omega}\frac{v_{0,0}}{U_{0,0}}(\rho_{n,m}/\rho_{0,0}-1)}.
\end{equation}
Now redefining the constant $\rho_{0,0}$ we can also write this as
\begin{equation}\label{v-res}
v_{n,m}=U_{nm}\,\tfrac{T-\Omega}{\tau}\,\frac{\rho_{n,m}}
{1+\rho_{n,m}}.
\end{equation}

\subsection{N-soliton solutions and Casoratians}
\label{sec:cas}
Having a 1SS in the form mentioned above allows us to propose a change
of dependent variables such that the original equation is given in
terms of discrete equations in HB form. These equations will then be
shown to have solutions given in Casorati determinant form,
corresponding to the Wronskian form solutions for continuous HB
equations. We will now discuss some generalities about the
Casoratians.

In general the Casorati matrix is constructed as follows: given
functions $\psi_i(n,m,l)$ we define the column vector
\begin{subequations}
\begin{equation}\label{C-entry-vec}
\psi(n,m,l)=(\psi_1(n,m,l),\psi_2(n,m,l),\cdots,\psi_{N}(n,m,l))^T,
\end{equation}
and then the generic $N\times N$ Casorati matrix is composed of such
columns with different shifts $l_i$, with the determinant
\begin{equation}
  \label{eq:C-gen}
  C_{n,m}(\psi;\{l_i\})=
  |\psi(n,m,l_1),\psi(n,m,l_2),\cdots,\psi(n,m,l_N)|.
\end{equation}
Two typical  Casoratians that will be used later are
\begin{eqnarray}
  C_{n,m}^1(\psi)&:=&|\psi(n,m,0),\psi(n,m,1),\cdots,\psi(n,m,N-1)|
\nonumber\\
  &\equiv&|0,1,\cdots,N-1| \equiv  |\h{N-1}|,
\label{Caso-H1-1}\\
  C_{n,m}^2(\psi)&:=&|\psi(n,m,0),\cdots,\psi(n,m,N-2),\psi(n,m,N)|
\nonumber\\
  &\equiv&|0,1,\cdots,N-2,N| \equiv  |\h{N-2},N|,
\label{Caso-H1-2}
\end{eqnarray}
\end{subequations}
where we have also introduced the standard shorthand notation
\cite{Freeman-Nimmo-KP} in which only the shifts are given, and where
furthermore sequential changes in the column index are indicated by a
hat: $|0,\cdots,M,\dots|\equiv|\h{M},\dots|$.
We also denote $|\t{N-1}|=|1,2,\cdots, N-1|$.

In the Casoratians used in this work the entries $\psi_i$ in the
$N$th-order vector \eqref{C-entry-vec} are given by
\begin{equation}
\psi_i(n,m,l)= \varrho_{i}^{+}(a +k_i)^n(b +k_i )^m(c+ k_i)^l+
  \varrho_{i}^{-}(a -k_i)^n(b -k_i )^m(c-k_i)^l,
  \label{psi-gen}
\end{equation}
where $\varrho_{i}^{\pm}$ and $k_i$ are parameters (and in some cases
$c=0$).  In this form the shifts in $l$ (bar-shifts) are not in any way
special and therefore we can equally well define Casoratians where the
shifts are in $n$ (tilde-shifts) or in $m$ (hat-shifts).  For later use
we indicate these three different shifts by the shift operators
$E^{\nu}$, $\nu=1,2,3$, respectively, i.e.,
\[
E^{1}\psi\equiv \t\psi,\quad E^{2}\psi\equiv \h\psi,\quad
E^{3}\psi\equiv \b\psi.
\]
Down shifts are denoted by $E_{\nu}$, $\nu=1,2,3$, respectively.

We can now define a Casoratian w.r.t. the three different kinds of shifts,
\begin{equation}
|\h{N-1}|_{_{[\nu]}}=|\psi, E^{\nu}\psi,(E^{\nu})^2\psi,\cdots,
(E^{\nu})^{N-1}\psi|,\quad (\nu=1,2,3).
\label{caso-def}
\end{equation}
Since
\begin{equation}
(\alpha_\mu-\alpha_\nu)\psi=(E^{\mu}-E^{\nu})\psi,\quad  \mu,\nu=1,2,3,
\end{equation}
where
\begin{equation}
\alpha_1\equiv a,\quad \alpha_2 \equiv b,\quad \alpha_3\equiv c,
\label{alpha-ij}
\end{equation}
we have
\begin{equation*}
(E^{\nu})^j\psi=[E^{\mu}-(\alpha_\mu-\alpha_\nu)]^j\psi,\quad  \mu,\nu=1,2,3,
\end{equation*}
and substituting this into \eqref{caso-def} yields
\begin{equation}
|\h{N-1}|_{_{[1]}}=|\h{N-1}|_{_{[2]}}=|\h{N-1}|_{_{[3]}}.
\end{equation}

When the Casoratians are constructed using $\psi$ of \eqref{psi-gen}
the size on the matrix $N$ indicates the number of solitons and the
set $\{k_i\}_{i=1}^N$ provides the ``velocity'' parameters of the
solitons, while the parameters $\{\rho_{i}^{+},\rho_i^{-}\}_{1}^N$ are
related to the locations of the solitons (by gauge invariance only
their ratio is significant).

We shall here mention that the bilinear equations we get in the paper
are more or less similar to the Hirota-Miwa equation\cite{Z-1,Z-2}
\begin{equation}
  a(b-c)\tau_{n,m,l+1}\tau_{n+1,m+1,l}+b(c-a)\tau_{n,m+1,l}\tau_{n+1,m,l+1}
+c(a-b)\tau_{n+1,m,l}\tau_{n,m+1,l+1}=0,
\label{H-W eq}
\end{equation}
(cf.(\ref{eq:bil-HM},\ref{eq:bil-H3-2M})), or belong to the bilinear
equations\cite{Z-3} which were derived by imposing the
transformation\cite{Z-2}
\begin{equation}
x_j=\sum^{\infty}_{i=1}l_j\frac{a^j_i}{j}
\label{Miwa-trans}
\end{equation}
on Sato's bilinear identity.  The above transformation, referred to as
Miwa transformation, provides a connection between continuous
coordinates $\{ x_j\}$ and discrete ones $\{ l_j\}$, and transforms
the basic continuous plane wave factor
\begin{equation}
\exp\Bigl[{-\sum^{\infty}_{jx=1}x_jp_i+\sum^{\infty}_{j=1}x_jq_i}\Bigr]
\end{equation}
into the discrete one
\begin{equation}
\prod^{\infty}_{j=1}\Bigl(\frac{1-a_j p_i}{1-a_jq_i}\Bigr)^{l_j}.
\label{dis-e}
\end{equation}
The discrete exponential function \eqref{dis-e} corresponding to the
plane wave factor $\rho_i$ (I-2.2), plays a central role in the
discrete $\tau$ function in Hirota's exponential-polynomial form,
while in Casoratians its counterpart has the form \eqref{psi-gen}.

In most cases a bilinear equation that can be solved by an $N\times
N$ Casoratian (or Casoratians) is reduced to a Laplace
expansion of a $2N\times 2N$ determinant with zero value.  For later
convenience we give the following lemmas.
\begin{lem}\cite{Freeman-Nimmo-KP}
\label{L:iden}
\begin{equation}
\sum_{j=1}^{N}|\mathbf{a}_1,\cdots,\mathbf{a}_{j-1},\,
\mathbf{b}\mathbf{a}_{j},\, \mathbf{a}_{j+1},\cdots,
\mathbf{a}_{N}|=\biggl(\sum_{j=1}^{N}b_{j}\biggr)|\mathbf{a}_1,\cdots,
 \mathbf{a}_N|,
\label{iden-1}
\end{equation}
where $\mathbf{a}_j=(a_{1j},\cdots,a_{Nj})^T$ and
$\mathbf{b}=(b_1,\cdots,b_N)^T$ are $N$-order column vectors, and
$\mathbf{b}\mathbf{a}_j$ stands for $(b_{1}a_{1j}, \cdots, b_N
a_{Nj})^{T}$.
\end{lem}

\begin{lem}\cite{Freeman-Nimmo-KP}
\label{L:lap}
Suppose that $\mathbf{B}$ is an $N\times(N-2)$ matrix and
$\mathbf{a},~ \mathbf{b},~ \mathbf{c},~ \mathbf{d}$ are $N$-order
column vectors, then
\begin{equation}
|\mathbf{B},\mathbf{a}, \mathbf{b}||\mathbf{B}, \mathbf{c},\mathbf{d}|
-|\mathbf{B}, \mathbf{a}, \mathbf{c}||\mathbf{B},\mathbf{b},\mathbf{d}|
+|\mathbf{B},\mathbf{a},\mathbf{d}||\mathbf{B},\mathbf{b},\mathbf{c}|=0.
\label{laplace}
\end{equation}
\end{lem}

In fact, the l.h.s. of \eqref{laplace} is just the
Laplace expansion of the following $2N\times 2N$ determinant,
$$
\frac12\left|
\begin{array}{cccccc}
\mathbf{B} & \mathbf{0} & \mathbf{a} & \mathbf{b} & \mathbf{c} & \mathbf{d} \\
\mathbf{0} & \mathbf{B} & \mathbf{a} & \mathbf{b} & \mathbf{c} & \mathbf{d}
\end{array}
\right|\equiv 0.
$$

\section{H1}
\subsection{Background solution and 1-soliton solution}
\subsubsection{The background solution}
The H1 equation is given by
\begin{equation}
{\rm H}1 \equiv (u-\th{u})(\t{u}-\h{u})-(p-q)=0.
\label{H1}
\end{equation}
Using the fixed point idea with transformation $\b u=T(u)=u+c$
we get the side-equations of the  CAC-cube in the form
\[
(\t u-u)^2=r+c^2-p,\quad (\h u-u)^2=r+c^2-q.
\]
For convenience we absorb $r$ into $c^2$, and reparameterize
$(p,q)\to(a,b)$ by
\begin{equation}
  \label{H1-repar}
  p=c^2-a^2,\quad q=c^2-b^2.
\end{equation}
and then the above equations factorize as
\begin{equation}
  \label{H1-0ss-eqs}
  (\t{u}-u-a)(\t{u}-u+a)=0,
\quad
  (\h{u}-u-b)(\h{u}-u+b)=0,
\end{equation}
Since the factor that vanishes may depend on $n,m$ we actually
have to solve
\begin{equation}
  \label{A-lin2eq-2}
  \t{u}-u=(-1)^{\theta}\,a,
\quad
  \h{u}-u=(-1)^{\chi}\,b,
\end{equation}
where $\theta,\,\chi\in\mathbb Z$ may depend on $n,m$.  Furthermore,
since the value of the exponent is only relevant modulo $2$ and
$n^2\equiv n \mod 2$, the exponents must be linear combinations of
$1,n,m,nm$ with coefficients $0$ or $1$.

The integrability condition for \eqref{A-lin2eq-2} leads to
\[
b\left[(-1)^\sigma-(-1)^{\t \sigma}\,\right]=
a\bigl[(-1)^\theta-(-1)^{\h \theta}\,\bigl],
\]
and since $a,b$ are independent this just means that
$\sigma=\sigma(m), \,\theta=\theta(n)$. Since these are defined modulo
2 we can have $\theta=t_1 n+t_0,\,\sigma=s_1 m+s_0$ where
$t_i,s_i\in\{0,1\}$. Furthermore, since $a,b$ were only defined up to
sign we may take $t_0=s_0=0$.

There are now two essentially different cases for each exponent,
either $\theta(n)=0$ or $n$, and $\sigma(m)=0$ or $m$. The choice
$\theta=0$ leads to a solution $u=a n +\dots$ while $\theta=n$ leads
to $u=-\tfrac12 (-1)^n a + \dots$. These solutions can be combined,
and thus we get the following set of possibilities for $u^{0SS}$
\begin{subequations}\label{A-lin2-0ss}
\begin{align}
  \label{A-lin2-0ssll}
a n+b m + \gamma,\\
  \label{A-lin2-0ssel}
\tfrac12(-1)^{n}a+bm+\gamma,\\
  \label{A-lin2-0ssle}
an+\tfrac12(-1)^m b+\gamma,\\
  \label{A-lin2-0ssee}
\tfrac12(-1)^{n}a+\tfrac12(-1)^m b+\gamma.
\end{align}
\end{subequations}

We could also try to get a background solution using the fixed point
idea with the other transformation mentioned for H1 in Table \ref{T:inv},
namely $\b u=T(u)=-u+c$, but it produces the same set: with a
reparameterization
\begin{equation}
  \label{H1m-repar}
  p=r+\alpha^2,\quad q=r+\beta^2,\quad u=y+\tfrac12c.
\end{equation}
the side-equations become
\begin{equation}
  \label{H1m-0ss-eqs}
  (\t{y}-y-\alpha)(\t{y}-y+\alpha)=0,\quad
  (\h{y}-y-\beta)(\h{y}-y+\beta)=0,
\end{equation}
which is as in \eqref{H1-0ss-eqs}. The constant $c$ can be absorbed
into the constant $\gamma$ and the free parameter $r$ redefined when
comparing \eqref{H1-repar},\eqref{H1m-repar}, and thus the background
solutions are as in \eqref{A-lin2-0ss}.

It seems that only \eqref{A-lin2-0ssll} leads to a soliton type
solution so in the following we only consider it.

\subsubsection{Constructing 1SS using a BT}
The BT generating 1SS for H1 is
\begin{subequations}
\label{BT-H1}
\begin{equation}
(u-\tb{u})(\t{u}-\b{u})=p-\varkappa,
\label{BT-H1-a}
\end{equation}
\begin{equation}
(u-\hb{u})(\b{u}-\h{u})=\varkappa-q.
\label{BT-H1-b}
\end{equation}
\end{subequations}
Here $u$ is the seed solution \eqref{A-lin2-0ssll}, $\varkappa$ is the
parameter in the bar-direction, and we search for a
new solution $\b{u}$ of the form
\begin{equation}
\b{u}=\b{u}_{0}+v,
\label{u-bar-H1}
\end{equation}
where $\b{u}_{0}$ is the bar-shifted background solution \eqref{A-lin2-0ssll}:
\begin{equation}
\b{u}_{0}=a n +b m +k+\gamma,
\label{u-b-theta-H1}
\end{equation}
where $k$ is related to $\varkappa$ by
\begin{equation}
\varkappa=c^2-k^2.
\label{s-k-H1}
\end{equation}

Substituting \eqref{A-lin2-0ssll}, \eqref{u-bar-H1} and
\eqref{u-b-theta-H1} into \eqref{BT-H1} yields the
desired form
\begin{equation}
\t{v}=\frac{Ev}{v+F},\quad\h{v}=\frac{Gv}{v+H},
\label{v-H1-1}
\end{equation}
where
\begin{equation*}
E=-(a+k),\quad F=-(a-k),\quad G=-(b +k),\quad H=-(b-k).
\end{equation*}

Now matrices $\mathcal{N,M}$ defined in \eqref{NM-matricesEG} are
$n,m$ independent and since $E-F=G-H=-2k$ they commute and we can take
$\Lambda=\Lambda'=1$. Then, following the method presented in Section
\ref{S:1ss}, we find
\begin{equation}
\Phi(n,m)=
\left(\begin{array}{cc}
      E^nG^m               & 0\\
      \frac{E^nG^m-F^nH^m}{-2k}     & F^nH^m
      \end{array}
\right)\Phi(0,0),
\label{Phi-nm-H1}
\end{equation}
and if we let
\begin{equation}\label{H1-1ss}
\rho_{n,m}=\biggl(\frac{E}{F}\biggr)^n\biggl(\frac{G}{H}\biggr)^m \rho_{0,0}
=\biggl(\frac{a +k}{a -k}\biggr)^n\biggl(\frac{b + k}{b-k}\biggr)^m \rho_{0,0},
\end{equation}
then we obtain
\begin{equation}\label{H1-v1ss}
v_{n,m}=\frac{-2k \rho_{n,m}}{1+\rho_{n,m}}.
\end{equation}
Finally we obtain the 1SS for H1:
\begin{equation}
u_{n,m}^{1SS}=(an+bm+\gamma)+k+\frac{-2k \rho_{n,m}}{1+\rho_{n,m}}.
\label{u1ss-H1}
\end{equation}

\subsection{Multi-soliton solutions}
In an explicit form the 1SS given above [\eqref{u1ss-H1} with
\eqref{H1-1ss}] is
\begin{equation}\label{1SS-H1}
\begin{array}{rl}
  u_{n,m}^{1SS} & = a n + b m +\gamma +
  \frac{k(1-\rho_{n,m})}{1+\rho_{n,m}}\\&
 =  a n + b m +\gamma
  -\frac{k[\varrho^{+}(a +k)^n(b +k)^m -\varrho^{-}(a - k)^n(b - k)^m]}
  {\varrho^{+}(a +k)^n(b +k )^m+ \varrho^{-}(a -k)^n(b -k)^m},
\end{array}
\end{equation}
where we have separated $\rho_{0,0}$ to $\varrho^-/\varrho^+$.

The above suggests that the basic ingredient in constructing an NSS is
$\psi$ of \eqref{psi-gen} with $c=0$, i.e.,
\begin{equation}
  \psi_i(n,m,l;k_i)=
  \varrho_{i}^{+}k_i^l (a +k_i)^n(b +k_i )^m+
  \varrho_i^{-}(-k_i)^l (a -k_i)^n(b -k_i)^m.
\label{entry-H1-1}
\end{equation}
Here the index $l$ is needed as a column index in the
Casorati matrix.

The form of the 1SS \eqref{1SS-H1} suggests a generalization to NSS,
and after checking explicitly the 2SS and 3SS, we arrive at the following
proposition:
\begin{prop}\label{P:H1-1}
$N$-soliton solution for H1 is given by
\begin{equation}
  u_{n,m}^{NSS}  = a n + b m +\gamma
  - \frac{g}{f},
\label{uNss-H1}
\end{equation}
where
$f=|\h{N-1}|_{_{[3]}},\,g=|\h{N-2},N|_{_{[3]}}$, with $\psi$
given by \eqref{entry-H1-1}.
\end{prop}

Note that this solution is the same as (5.29) of Part I with $A=0$.
Indeed using (I-2.1,2.3,2.8c) one finds that
\[
S^{(0,0)}-\sum k_i = \frac gf,
\]
and the parameters $\varrho_i^{\pm}$ in (\ref{entry-H1-1}) and $c_i$ in (I-2.3) are
related by
\begin{equation}\label{par-comp}
\frac{\varrho_i^+}{\varrho_i^-}=\frac{c_i}{2k_i}\prod_{j\neq i}\frac{k_j-k_i}{k_j+k_i}.
\end{equation}
However, the solution \eqref{uNss-H1} contains more freedom in the
form of the additional parameter $c$, cf. \eqref{H1-repar} and
(I-5.1d).

\subsubsection{The bilinear form of H1}
In order to prove Proposition \ref{P:H1-1} we write \eqref{H1} in
bilinear form. First we introduce a dependent variable transformation
\begin{equation}
  u_{n,m}^{NSS}  = a n + b m +\gamma
  - \frac{g_{n,m}}{f_{n,m}}.
\label{depvar-H1}
\end{equation}
When this is substituted into \eqref{H1} we get a rational expression
quartic in $f,g$. In order to split this expression into bilinear
equations we note that if $f=|\h{N-1}|_{_{[3]}}$, $g=|\h{N-2},N|_{_{[3]}}$, then for
$N=2$ one can quickly scan for possible equations solved by them from
the set
\[
a_0g\th f+a_1\h g\t f+a_2 \t g\h f+a_3\th g f+
a_4f\th f+a_5\h f\t f+a_6 g\th g+a_7\h g\t g=0,
\]
and find the bilinear equations
\begin{subequations}  \label{eq:bil-H1}
\begin{eqnarray}
\mathcal{H}_1 &\equiv & \h g\t f-\t g\h f+(a-b)(\h f\t f- f\th f)=0,
\label{eq:bil-H1-1}\\
\mathcal{H}_2 &\equiv & g\th f-\th g f+(a+b)(f\th f-\h f\t f)=0
\label{eq:bil-H1-2}.
\end{eqnarray}
\end{subequations}
After this it is easy to show that
\[
\mathrm{H}1\equiv (u-\th{u})(\t{u}-\h{u})-p+q
=-\bigl[\mathcal{H}_1+(a-b)f\th f\bigr]
\bigl[\mathcal{H}_2+(a+b)\h f\t f\bigr]/(f\h f\t f \th f)
+(a^2-b^2),
\]
and thus the pair \eqref{eq:bil-H1} can be taken as the bilinear form of
\eqref{H1}.

\subsubsection{Proof of the Proposition \ref{P:H1-1}}
Now it remains to show that the $f,g$ given in Proposition
\ref{P:H1-1} solve equations \eqref{eq:bil-H1}.  We prove
\eqref{eq:bil-H1-1} in its tilde-hat-downshifted version,
\begin{equation}
  \label{eq:bil-H1-1d}
\dt g\dh f-\dh g\dt f +(a-b)(\dh f\dt f- f\dth f)=0.
\end{equation}
Comparing $\psi_i(n,m,l)$ given by \eqref{entry-H1-1} and
\eqref{entry-H3-1}, we can use the formulas given in Appendix
\ref{A:a} with $c\equiv 0$, and in this subsection we always have
shifts in the $l$-index, i.e., we take $\kappa=3$ for the formulas in
Appendix \ref{A:a}.  Let us write \eqref{eq:bil-H1-1d} as
\begin{equation}
  \label{eq:bil-H1-1d-al}
  -(a-b)f\dth f + \dh f (\dt g+a\dt f) -\dt f(\dh g + b\dh f)=0.
\end{equation}
In this formula, $f=|\h{N-1}|_{_{[3]}}$, and for $\dth f$, $\dh f$,
$\dt g+a\dt f$, $\dt f$ and $\dh g + b\dh f$ we use \eqref{I-k} with
$\mu=1$ and $\nu=2$, \eqref{I-a} with $\mu=2$, \eqref{I-g-a} with
$\mu=1$, \eqref{I-a} with $\mu=1$ and \eqref{I-g-a} with $\mu=2$,
respectively. Then we have
\begin{eqnarray*}
  &~& a^{N-2}b^{N-2}[-(a-b)f\dth f + \dh f (\dt g+a\dt f)
 -\dt f(\dh g + b\dh f)]\\
  &=& -|\h{N-1}|_{_{[3]}} |\h{N-3},\dh{\psi}(N-2),\dt{\psi}(N-2)|_{_{[3]}}\\
  &~& +|\h{N-2},\dh{\psi}(N-2)|_{_{[3]}} |\h{N-3},N-1,\dt{\psi}(N-2)|_{_{[3]}}\\
  &~& -|\h{N-2},\dt{\psi}(N-2)|_{_{[3]}} |\h{N-3},N-1,\dh{\psi}(N-2)|_{_{[3]}} \\
  &=& 0,
\end{eqnarray*}
where we have made use of Lemma \ref{L:lap} in which
$\mathbf{B}=(\h{N-3})$, and $\mathbf{a}=\psi(N-2)$,
$\mathbf{b}=\psi(N-1)$, $\mathbf{c}=\dh{\psi}(N-2)$ and
$\mathbf{d}=\dt{\psi}(N-2)$.

For \eqref{eq:bil-H1-2} we use its tilde-down-shifted version
\begin{equation}
\dt g\h f-\h g\dt f+(a+b)(\dt f\h f-\dt{\h f}f)=0,
\label{eq:bil-H1-2d}
\end{equation}
and rewrite it as
\begin{equation*}
-(a+b) f \dt{\h f} +\h f (\dt g +a\dt f) -\dt f (\h g -b \h f)=0,
\end{equation*}
where $f=|\h{N-1}|_{_{[3]}}$, and for $\dt{\h f}$, $\h f$, $\dt g+a\dt
f$, $\dt f$ and $\h g - b\h f$ we use \eqref{I-l} with $\mu=2$ and
$\nu=1$, \eqref{I-c} with $\mu=2$, \eqref{I-g-a} with $\mu=1$,
\eqref{I-a} with $\mu=1$ and \eqref{I-g-b} with $\mu=2$, respectively.
Then we have
\begin{eqnarray*}
  &~& \frac{1}{|\Omega_2|}a^{N-2}b^{N-2}[-(a+b) f \dt{\h f}
 +\h f (\dt g +a\dt f) -\dt f (\h g -b \h f)]\\
  &=& -|\h{N-1}|_{_{[3]}} |\h{N-3},\dt{\psi}(N-2),
\c E^{2}{\psi}(N-2)|_{_{[3]}}\\
  &~& -|\h{N-2},\c E^{2}{\psi}(N-2)|_{_{[3]}} |\h{N-3},N-1,
\dt{\psi}(N-2)|_{_{[3]}}\\
  &~& +|\h{N-2},\dt{\psi}(N-2)|_{_{[3]}} |\h{N-3},N-1,
\c E^{2}{\psi}(N-2)|_{_{[3]}} \\
  &=& 0,
\end{eqnarray*}
where we have made use of Lemma \ref{L:lap} in which
$\mathbf{B}=(\h{N-3})$, and $\mathbf{a}=\psi(N-2)$,
$\mathbf{b}=\psi(N-1)$, $\mathbf{c}=\dt{\psi}(N-2)$ and $\mathbf{d}=\c
E^{2}{\psi}(N-2)$. \qed

\section{H2}
\subsection{Background solution}
The equation H2 is given by
\begin{equation}
\mathrm{H}2\equiv (u-\th{u})(\t{u}-\h{u})-(p-q)(u+\t{u}+\h{u}+\th{u}+p+q)=0.
\label{H2}
\end{equation}

After reparameterization
\begin{equation}
  \label{H2-repar}
  p=r-a^2,\quad q=r-b^2,
\end{equation}
the equations on sides become
\begin{equation}
  \label{H2-00s-eqs}
(\t u-u)^2-2a^2(\t u+u)+a^2(a^2-2r)=0,\quad
(\h u-u)^2-2b^2(\h u+u)+b^2(b^2-2r)=0.
\end{equation}
After the further substitution
\begin{equation}
  \label{H2-x-sub}
  u=y^2-\tfrac12r
\end{equation}
the equations \eqref{H2-00s-eqs} factorize as
\begin{subequations}
\label{H2-fac}
\begin{eqnarray}
(\t y +y+a)(\t y +y-a)(\t y -y+a)(\t y -y-a)=0,\\
(\h y +y+b)(\h y +y-b)(\h y -y+b)(\h y -y-b)=0.
\end{eqnarray}
\end{subequations}
Thus we have to solve
\begin{equation}\label{A2-2signs}
\t y-(-1)^\sigma\, y=(-1)^\theta\, a, \quad
\h y-(-1)^\rho\, y=(-1)^\phi\, b.
\end{equation}
where the exponents $\sigma,\theta,\rho,\phi\in\mathbb Z$ are again
some linear combination of $1,n,m,nm$ with coefficients $0$ or
$1$.  Equations \eqref{A2-2signs} are satisfied if
\[
\rho=n\, s_2+r_1,\quad \sigma=m\, s_2+s_1,\quad
\theta=m\,\rho+\sigma+n(t_n+s_1)+t_1,\quad
\phi=n\,\sigma+\rho+m(p_m+r_1)+p_1,
\]
but most of the freedom is superfluous: Since the sign of $y$ was left
undetermined in \eqref{H2-x-sub} we can redefine
\[
y_{n,m}\to (-1)^{nm\,s_2+n\,s_1+m\,r_1}\, y_{n,m}
\]
and then the equations simplify to
\[
\t y-y=(-1)^{n\,t_n+t_1}a,\quad
\h y-y=(-1)^{m\,p_m+p_1}b,
\]
already analyzed after \eqref{A-lin2eq-2}. Thus we find the possible
background solutions $u^{0SS}_{n,m}$ of the form
\begin{subequations}\label{H2-0ss}
\begin{eqnarray}
  \label{H2-0ssll}
 &&[a n+b m + \gamma]^2-\tfrac12r,\\
  \label{H2-0ssel}
&&[\tfrac12(-1)^{n}a+mb+\gamma]^2-\tfrac12r,\\
  \label{H2-0ssle}
&&[na+\tfrac12(-1)^mb+\gamma]^2-\tfrac12r,\\
  \label{H2-0ssee}
&&[\tfrac12(-1)^{n}a+\tfrac12(-1)^mb+\gamma]^2-\tfrac12r.
\end{eqnarray}
\end{subequations}

\subsection{1-soliton solution}
Note that H2 is invariant under simultaneous translation of $p,q$ by
$t$ and $u$ by $-\tfrac12 t$. We use this freedom to eliminate $r$ in
order to simplify the presentation.

Now we take \eqref{H2-0ssll} as the seed solution to construct
$u_{n,m}^{1SS}$ for H2 through its B\"acklund transformation
\begin{subequations}
\label{BT-H2}
\begin{equation}
(u-\tb{u})(\t{u}-\b{u})=(p-\varkappa)(u+\t{u}+\b{u}+\tb{u}+p+\varkappa),
\label{BT-H2-a}
\end{equation}
\begin{equation}
(u-\hb{u})(\b{u}-\h{u})=(\varkappa-q)(u+\b{u}+\h{u}+\bh{u}+\varkappa+q).
\label{BT-H2-b}
\end{equation}
\end{subequations}
where $u$ is the seed solution \eqref{H2-0ssll} and we search for a new
solution $\b{u}$ of the form
\begin{equation}
\b{u}=\b{u}_{0}+v,
\label{u-bar-H2}
\end{equation}
where $\b{u}_{0}$ is the bar-shifted background solution \eqref{H2-0ssll}:
\begin{equation}
\b{u}_{0}=(a n +b m +k+\gamma)^2,
\label{u-b-theta-H2}
\end{equation}
with
\begin{equation}
\varkappa=k^2.
\label{s-k-H2}
\end{equation}
Substituting these
into \eqref{BT-H2} yields \eqref{NM-matricesEG} with
\begin{eqnarray*}
E=-2(k+a)[(n+1)a+mb+\gamma],\quad
F=-2(k-a)[na+mb+\gamma],\\
G=-2(k+b)[na+(m+1)b+\gamma],\quad
H=-2(k-b)[na+mb+\gamma].
\end{eqnarray*}
and the corresponding matrices $\mathcal{N,M}$ are compatible if we
take
\begin{equation}
\Lambda=\Lambda'=-1/(2U_{nm}),\quad U_{n,m}:=a n +b m +\lambda.
\label{U-H2}
\end{equation}
With this $U$ we obtain  \eqref{NM-matricesST} with
\begin{equation}
S=a+k,\quad \Delta=a -k,\quad T=b+k,\quad  \Omega=b-k,\quad \sigma=\tau=-1/2.
\end{equation}
Then defining
\begin{equation}\label{H2-1ss}
\rho_{n,m}=\biggl(\frac{S}{\Delta}\biggr)^n\biggl(\frac{T}{\Omega}\biggr)^m \rho_{0,0}
=\biggl(\frac{a+ k}{a -k}\biggr)^n\biggl(\frac{b + k}{b-k}\biggr)^m \rho_{0,0},
\end{equation}
we find
\begin{equation}
v_{n,m}=\frac{-4k U_{n,m} \rho_{n,m}}
{1+\rho_{n,m}},
\end{equation}
and finally we get  the 1-soliton for H2 in the form
\begin{equation}
u_{n,m}^{1SS}  = U_{n,m}^2
+2k U_{n,m}\frac{1-\rho_{n,m}}{1+\rho_{n,m}} +k^2.
\label{1ss-H2}
\end{equation}

\subsection{Multi-soliton solution}
Motivated by the structure of 1SS \eqref{1ss-H2} (cf. \eqref{1SS-H1}),
and after checking 2- and 3-soliton solutions we propose the following
Casoratian expression for the NSS:
\begin{equation}
u_{n,m}^{NSS}  =U_{n,m}^2
-2 U_{n,m} \frac{|\h{N-2},N|_{_{[3]}}}{|\h{N-1}|_{_{[3]}}}
+\frac{|\h{N-3},N-1,N|_{_{[3]}}+|\h{N-2},N+1|_{_{[3]}}}{|\h{N-1}|_{_{[3]}}},
\label{Nss-H2}
\end{equation}
where the matrix entries are as for H1 \eqref{entry-H1-1}. In order to
prove this we derive a bilinear form of H2. We propose
\begin{equation}
u_{n,m}^{NSS}  =U_{n,m}^2
-2 U_{n,m} \frac{g}{f}
+\frac{h+s}{f},
\label{trans-H2}
\end{equation}
where $f$, $h$ and $s$ should satisfy
\begin{equation}
h-s=\gamma f,
\label{cond-fhs}
\end{equation}
where $\gamma$ is some constant.  Indeed, if
\begin{equation}\label{H2-casdef}
  f=|\h{N-1}|_{_{[3]}},\,  g=|\h{N-2},N|_{_{[3]}},
\, s=|\h{N-3},N-1,N|_{_{[3]}},\,  h=|\h{N-2},N+1|_{_{[3]}},
\end{equation}
then noting that
$(E^3)^2\psi_i(n,m,l)=\bb\psi_i(n,m,l)=\psi_i(n,m,l+2)=k_i^2
\psi_i(n,m,l)$ and using Lemma \ref{L:iden} we have
$\bigl(\sum_{i=1}^{N}k^2_{i}\bigr) f=h-s,$
i.e. $\gamma=\sum_{i=1}^{N}k^2_{i}$.

The solution \eqref{trans-H2} is the same as (I-5.26) with $A=0$ and
\[
S^{(0,0)} = {\textstyle \sum_j} k_j + \frac{g}{f},\quad
2S^{(0,1)} -2({\textstyle \sum_j} k_j)S^{(0,0)}+ ({\textstyle \sum_j}
k_j)^2 = \frac {h+s}f,
\]
which can be found by using (I-2.1,2.2,2.3,2.8c,5.23).

Under condition \eqref{cond-fhs}, H2 can be represented through the following
bilinear system,
\begin{subequations}
\label{eq:bil-H2}
\begin{eqnarray}
\mathcal{H}_1 &\equiv & \h g\t f-\t g\h f+(a-b)(\h f\t f- f\th f)=0,
 \label{eq:bil-H2-1}\\
\mathcal{H}_2 &\equiv & g\th f-\th g f+(a+b)(f\th f-\h f\t f)=0,
 \label{eq:bil-H2-2}\\
\mathcal{H}_3 &\equiv & -(a+b)\h{f}\t{g} + a\th{f} g +b f \th{g}+
\th f h -f \th{h}=0,\label{eq:bil-H2-3}\\
\mathcal{H}_4 &\equiv & -(a-b)f\th{g} + a \t{f}\h{g}-b\h{f} \t{g} +
\t{f}\,\h{h}-\h f\, \t{h}=0,\label{eq:bil-H2-4}\\
\mathcal{H}_5 &\equiv & b(\h f g - f \h{g}) + f \h h+\h f s -g\h{g}=0,
\label{eq:bil-H2-5}
\end{eqnarray}
\end{subequations}
in which $\mathcal{H}_1$ and $\mathcal{H}_2$ already appeared in
\eqref{eq:bil-H1}.  In terms of the above bilinear equations H2 can
be given as
\begin{equation}
\mathrm{H}2=\sum^{5}_{i=1}\mathcal{H}_iP_i,
\end{equation}
with
\begin{equation*}
\begin{array}{rl}
  P_1=& -4(a+b)\Bigl[(\t U\th U-a^2+ b^2)\t f \h f-
\th U \h f\t g-(a+b)f\th g \Bigr],\\
  P_2=& -4\Bigl[(a-b)(\h U\th U-a^2+ b^2)\t f \h f
                      +(\t U\th U-a^2+ b^2)\t f \h g -
\t U\th U \h f\t g-(a-b)\t U f\th g \Bigr],\\
  P_3=& 4\Bigl[(a-b)U\t f\h f+\h U\t f\h g-\t U\h f\t g-
\t f\, \h h+\h f\,\t h \Bigr],\\
  P_4=& 4\Bigl[(a+b)(\h U f\th f-\h f \t g) +\t U(\th f g -f\th g)\Bigr],\\
  P_5=& 4(a^2-b^2)\t f\th f,
\end{array}
\end{equation*}
where $U$ was defined in \eqref{U-H2}. It remains to prove the following
\begin{prop}
\label{P:H2-Nss}
  The Casoratian type determinants $f$, $g$, $h$ and $s$
given in \eqref{H2-casdef}
with entries given by \eqref{entry-H1-1} solve the set of bilinear
equations \eqref{eq:bil-H2}.
\end{prop}

\begin{proof}
  Among \eqref{eq:bil-H2} $\mathcal{H}_1$ and $\mathcal{H}_2$ have
  been proven before.  Next we prove \eqref{eq:bil-H2-3} in the
  following form
\begin{equation}
-(a+b)\dt {\h{f}} {g} + \h{f}(\dt h+ a \dt g) - \dt f(\h h- b \h{g})=0,
\label{dt-BH3}
\end{equation}
which is a down-tilde-shifted version of the original one.  Since
$\psi_i(n,m,l)$ given by \eqref{entry-H1-1} is just \eqref{entry-H3-1}
with $c=0$, we use the formulas given in Appendix \ref{A:a} with
$c\equiv 0$ and $\kappa=3$.  In \eqref{dt-BH3}
$g=|\h{N-2},N|_{_{[3]}}$, and for $\dt {\h{f}}$, $\h f$, $\dt h+a\dt
g$, $\dt f$ and $\h h-b\h g$ we use \eqref{I-l} with $\mu=2$ and
$\nu=1$, \eqref{I-c} with $\mu=2$, \eqref{I-h-a} with $\mu=1$,
\eqref{I-a} with $\mu=1$ and \eqref{I-h-b} with $\mu=2$, respectively.
Then we have
\begin{eqnarray*}
  &~& \frac{1}{|\Omega_2|}a^{N-2}b^{N-2}[-(a+b)\dt {\h{f}} {g}
 + \h{f}(\dt h+ a \dt g) - \dt f(\h h- b \h{g})]\\
  &=& -|\h{N-2},N|_{_{[3]}} |\h{N-3},\dt{\psi}(N-2),
\c E^{2}{\psi}(N-2)|_{_{[3]}}\\
  &~& -|\h{N-2},\c E^{2}{\psi}(N-2)|_{_{[3]}} |\h{N-3},N,
\dt{\psi}(N-2)|_{_{[3]}}\\
  &~& +|\h{N-2},\dt{\psi}(N-2)|_{_{[3]}} |\h{N-3},N,
\c E^{2}{\psi}(N-2)|_{_{[3]}} \\
  &=& 0,
\end{eqnarray*}
where we have made use of Lemma \ref{L:lap} in which
$\mathbf{B}=(\h{N-3})$, $\mathbf{a}=\psi(N-2)$, $\mathbf{b}=\psi(N)$,
$\mathbf{c}=\dt{\psi}(N-2)$ and $\mathbf{d}=\c E^{2}{\psi}(N-2)$

\eqref{eq:bil-H2-4} can be proved similarly after a down-tilde-hat-shift.

Next we prove \eqref{eq:bil-H2-5}, which can be written as
\begin{equation}
  f (\h h-b\h g)-g(\h g - b \h f) + \h f s =0,
\end{equation}
where $f=|\h{N-1}|_{_{[3]}},~
g=|\h{N-2},N|_{_{[3]}},~s=|\h{N-3},N-1,N|_{_{[3]}}$, and $\h h-b\h g$,
$\h g - b \h f$ and $\h f$ will be provided by \eqref{I-h-b} with
$\mu=2$, \eqref{I-g-b} with $\mu=2$ and \eqref{I-c} with $\mu=2$,
respectively.  Then we have
\begin{eqnarray*}
  &~& \frac{1}{|\Omega_2|}b^{N-2}[f (\h h-b\h g)-g(\h g - b \h f) + \h f s]\\
  &=& |\h{N-1}|_{_{[3]}} |\h{N-3},N,\c E^{2}{\psi}(N-2)|_{_{[3]}}\\
  &~& -|\h{N-2},N|_{_{[3]}} |\h{N-3},N-1,\c E^{2}{\psi}(N-2)|_{_{[3]}}\\
  &~& +|\h{N-2},\c E^{2}{\psi}(N-2)|_{_{[3]}} |\h{N-3},N-1,N|_{_{[3]}} \\
  &=& 0,
\end{eqnarray*}
where use has been made of Lemma \ref{L:lap} with
$\mathbf{B}=(\h{N-3})$, $\mathbf{a}=\psi(N-2)$,
$\mathbf{b}=\psi(N-1)$, $\mathbf{c}=\psi(N)$ and $\mathbf{d}=\c
E^{2}{\psi}(N-2)$.
\end{proof}

\section{H3}
\subsection{Background solution}
H$3^{\delta}$ is given by
\begin{equation}
\mathrm{H}3^\delta \equiv p(u\t{u}+\h{u}\th{u})-q(u\h{u}+\t{u}\th{u})
-\delta(q^2-p^2) =0.
\label{H3d}
\end{equation}
The side equations for $T(x)=x$ then read
\begin{equation}
  \label{H3-sides}
  r(u^2+\t u^2)-2pu\t u=\delta(p^2-r^2),\quad
  r(u^2+\h u^2)-2qu\h u=\delta(q^2-r^2).
\end{equation}

In this case we reparameterize
\begin{equation}
  \label{H3-reparx}
  p=r\cosh(\alpha'),\quad q=r\cosh(\beta'),\quad
  u_{nm}=Ae^{y_{nm}}+Be^{-y_{nm}},\quad AB=-\tfrac14 r \delta
\end{equation}
and then the equations \eqref{H3-sides} factorize as
\begin{subequations}
  \label{H3-exp-fact}
\begin{eqnarray}
  (e^{\t y-y+\alpha'}-1)  (e^{\t y-y-\alpha'}-1)
  (e^{\t y+y+\alpha'-\ln \frac{B}{A}}-1)
(e^{\t y+y-\alpha'-\ln \frac{B}{A}}-1)&=&0,\\
  (e^{\h y-y+\beta'}-1)  (e^{\h y-y-\beta'}-1)
  (e^{\h y+y+\beta'-\ln \frac{B}{A}}-1)
(e^{\h y+y-\beta'-\ln \frac{B}{A}}-1)&=&0.
\end{eqnarray}
\end{subequations}
Since we only consider real $u$ the various possibilities can be
represented as in \eqref{A2-2signs}.  The analysis is then the same,
especially since also here the sign of $y$ is undetermined in
\eqref{H3-reparx}. Thus the solution for $y$ is as in
\eqref{A-lin2-0ss} and for $u$ we have
\begin{eqnarray}
  u^{0SS}&=&Ae^{y_{nm}}+Be^{-y_{nm}}=A e^{\alpha' n+\beta' m+\gamma}+Be^{-\alpha'
    n-\beta' m-\gamma}\nonumber\\
  &=&A\alpha^n \beta^m+B\alpha^{-n}\beta^{-m},\quad AB=-\tfrac14 r
  \delta,
\label{u0ss-H3d}
\end{eqnarray}
where  $\alpha=e^{\alpha'}, \beta=e^{\beta'}$.

In the case of $T(x)=-x$, one can find that the side equations are
just as for $T(x)=x$ except that $r\to -r$. Since $r$ is a free
parameter this adds nothing new.

\subsection{1-soliton solution for H$3^{\delta}$}
Now we take \eqref{u0ss-H3d} as the seed solution to construct
$u_{n,m}^{1SS}$ for H$3^{\delta}$ through its B\"acklund
transformation
\begin{subequations}
\label{BT-H3}
\begin{equation}
p(u\t{u}+\b{u}\tb{u})-\varkappa(u\b{u}+\t{u}\tb{u})=\delta(\varkappa^2-p^2),
\label{BT-H3-a}
\end{equation}
\begin{equation}
\varkappa(u\b{u}+\h{u}\bh{u})-q(u\h{u}+\b{u}\bh{u})=\delta(q^2-\varkappa^2),
\label{BT-H3-b}
\end{equation}
\end{subequations}
where $u$ is the seed solution \eqref{u0ss-H3d}, and we search for the 1SS
$\b{u}$ of the form
\begin{equation}
\b{u}=\b{u}_{0}+v,
\label{u-bar-H3}
\end{equation}
where $\b{u}_{0}$ is the bar-shifted background solution \eqref{u0ss-H3d}:
\begin{equation}
\b{u}_{\theta}=A \alpha^n \beta^m \kappa +B \alpha^{-n}\beta^{-m}\kappa^{-1},
\label{u-b-theta-H3}
\end{equation}
with $\varkappa$ and $\kappa$ related by
\begin{equation}
\varkappa=r\frac{1+\kappa^2}{2\kappa}.
\label{s-k-H3}
\end{equation}

Following again the procedure in Section \ref{S:1ss} we get
$\mathcal{N,M}$ in the form \eqref{NM-matricesST} with
\begin{eqnarray*}
S=r\frac{1-\alpha^2 \kappa^2}{2\alpha \kappa},\quad
\Delta=r\frac{-\alpha^2+\kappa^2}{2\alpha \kappa},
\quad \sigma=p,\\
\quad
T=r\frac{1-\beta^2\kappa^2}{2\beta \kappa},\quad
\Omega=r\frac{-\beta^2+\kappa^2}{2\beta \kappa},\quad \tau=q,\\
U_{n,m}=A \alpha^n \beta^m -B \alpha^{-n}\beta^{-m}.
\end{eqnarray*}
Then defining $\rho$ by
\begin{equation}
\rho_{n,m}=\biggl(\frac{S}{\Delta}\biggr)^n
\biggl(\frac{T}{\Omega}\biggr)^m \rho_{0,0}
=\biggl(\frac{\alpha^2\kappa^2-1}{\alpha^2 -\kappa^2}\biggr)^n
\biggl(\frac{\beta^2\kappa^2-1}{\beta^2-\kappa^2}\biggr)^m \rho_{0,0},
\label{rho-H3}
\end{equation}
we find
\[
v_{n,m}=\frac{\frac{U_{n,m}}{U_{0,0}}v_{0,0}\rho_{n,m}/\rho_{0,0} }
{1-\frac{\kappa}{1-\kappa^2}\cdot\frac{v_{0,0}}{U_{0,0}}
  +\frac{\kappa}{1-\kappa^2}\cdot\frac{v_{0,0}}{U_{0,0}}\rho_{n,m}/\rho_{0,0}}
=\frac{\frac{1-\kappa^2}{\kappa} U_{n,m} \rho_{n,m}} {1+\rho_{n,m}},
\]
and finally
\begin{equation}\label{1ss-H3}
u_{n,m}^{1SS}  = \frac{A \alpha^n \beta^m(1+\kappa^{-2}\rho_{n,m})
+B \alpha^{-n}\beta^{-m}(1+ \kappa^2\rho_{n,m})}{1+\rho_{n,m}}.
\end{equation}

\subsection{Bilinear form and Casoratian solutions}
\subsubsection{$N$-soliton solution}
Noting that $\rho_{n,m}$ given by \eqref{rho-H3} is in a ``twisted''
form in comparison with \eqref{H1-1ss} and \eqref{H2-1ss}, we
first introduce the M\"obius transformations for the parameters
\begin{equation}
\alpha^2= -\frac{a-c}{a+c}, \,
\beta^2= -\frac{b-c}{b+c}, \,
\kappa^2= -\frac{k-c}{k+c},\quad
\Rightarrow\quad p^2=\frac{r^2c^2}{c^2-a^2},
\, q^2=\frac{r^2c^2}{c^2-b^2},
\label{Mob-trans}
\end{equation}
which also contain a new auxiliary parameter $c$.  This brings
$\rho_{n,m}$ of \eqref{rho-H3} into the canonical form
\begin{equation}
\rho_{n,m}=\biggl(\frac{a+k}{a-k}\biggr)^n
\biggl(\frac{b+k}{b-k}\biggr)^m \rho_{0,0}.
\label{rho-H3M}
\end{equation}
In terms of the new form of $\rho_{n,m}$ we can write the 1SS
\eqref{1ss-H3} as
\begin{equation}
u_{n,m}^{1SS} = A \alpha^n \beta^m \frac{\psi(n,m,l+1)}{\psi(n,m,l)}
+B \alpha^{-n}\beta^{-m}\frac{\psi(n,m,l-1)}{\psi(n,m,l)},
\quad AB=-\tfrac14r\delta,
\label{1ss-H3-m}
\end{equation}
where $\psi$ is as in  \eqref{psi-gen}.

On the basis of the above 1SS, we propose that the NSS of H3$^\delta$
can be given by
\begin{equation}
  u_{n,m}^{NSS}  = A \alpha^n \beta^m\frac{\b f}{f}+
  B \alpha^{-n}\beta^{-m}\frac{\db f}{f},\quad AB=-\tfrac14r\delta,
\label{trans-H3}
\end{equation}
where $f=|\h{N-1}|_{_{[\nu]}}$ with entries \eqref{psi-gen}.  We may
consider \eqref{trans-H3} as a dependent variable transformation for
\eqref{H3d}.

The solution \eqref{trans-H3} is the same as (I-5.21) with
$B=C=0$.  In fact, let $\frac 1r$ in \eqref{trans-H3} equal to $a$
which is the direction parameter for the bar-shift in Part-I.  Then
comparing \eqref{Mob-trans} with (I-5.1c), using (I-2.32,5.19)
and substituting $\rho_i$ in (I-2.2) by
\begin{equation}
\rho_i=\biggl(\frac{p+k_i}{p-k_i}\biggr)^n
\biggl(\frac{q+k_i}{q-k_i}\biggr)^m \biggl(\frac{a+k_i}{a-k_i}\biggr)^l \rho^{0}_i,
\label{rho-I}
\end{equation}
one finds
\[
\vartheta=\alpha^{-n}\beta^{-m},\quad
V(a)\,/\,{\textstyle \prod_j(a-k_j)}=\frac{\db f}{f},\quad
V(-a)\times{\textstyle \prod_j(a-k_j)}=\frac{\b f}{f},
\]
with the same parameter identification as in \eqref{par-comp}.

\subsubsection{Bilinearization-I}
After introducing the two bilinear equations
\begin{subequations}
\label{eq:bil-HM}
\begin{eqnarray}
\mathcal{B}_1 &\equiv & 2c f \t f +(a-c) \tb f \db f -(a+c)\b f \db{\t f} =0,
\label{eq:bil-HM-a}\\
\mathcal{B}_2 &\equiv & 2c f \h f +(b-c) \hb f \db f -(b+c)\b f \db{\h f} =0,
\label{eq:bil-HM-b}
\end{eqnarray}
\end{subequations}
we can represent H3$^\delta$ \eqref{H3d} as
\begin{equation*}
\mathrm{H}3^\delta \equiv \frac{-\alpha^{4n+2}\beta^{4m+2}(a+c)(b+c)\delta^2 P_1
+4\alpha^{2n}\beta^{2m}\delta B^2P_2+ 16 (a+c)(b+c)B^4 P_3}
{32 \alpha^{2n+2}\beta^{2m+2}(a+c)^2(b+c)^2 B^2 f\t f\h f \th f},
\end{equation*}
where
\begin{equation*}
\begin{array}{rl}
  P_1=& \th {\b f}\Bigl[(b-c)\bh f \mathcal{B}_1-(a-c)\bt f
{\mathcal{B}}_2 \Bigr]
        -\b f\Bigl[(b+c)\tb f \h {\mathcal{B}}_1-(a+c)\bh f
\t {\mathcal{B}}_2\Bigr],\\
  P_2=& 2c\Bigl[(b+c)(b-c)(\h f \th f \mathcal{B}_1+f \t f
\h {\mathcal{B}}_1)
        -(a+c)(a-c)(\t f\th f {\mathcal{B}}_2+f\h f
\t {\mathcal{B}}_2) \Bigr],\\
  P_3=& \th {\db f}\Bigl[(b+c)\db{\h f} \mathcal{B}_1-(a+c)
\db{\t f} {\mathcal{B}}_2\Bigr]
        -\db f\Bigl[(b-c)\db{\t f} \h {\mathcal{B}}_1-(a-c)
\db{\h f} \t {\mathcal{B}}_2\Bigr].
\end{array}
\end{equation*}
Thus \eqref{eq:bil-HM} can be considered as a bilinearization of
H3$^\delta$, and the final step in constructing the NSS is
\begin{prop}
\label{P:H3-Nss-1}
  The Casoratian
  \begin{equation}
  f=|\h{N-1}|_{_{[\nu]}},\quad  (\nu=1,2~{\rm or}~3),
  \end{equation}
  with entries given by $\psi$ of \eqref{psi-gen}, solves the bilinear
  H3$^{\delta}$ \eqref{eq:bil-HM}.
\end{prop}

\begin{proof}
We prove \eqref{eq:bil-HM-a} in its down-tilde-shifted version
\begin{eqnarray}
2c \dt f f +(a-c) \b f \dtb f -(a+c)\dt{\b f} \db f =0.
\label{eq:bil-H3-1Mad}
\end{eqnarray}
For this equation we use Casoratians w.r.t. hat shift, i.e.,
$\kappa\equiv 2$ in \eqref{Formula-I}.  $\dt f$, $f$, $\b f$, $\dtb
f$, $\dt{\b f}$ and $\db f$ are given by the formulas \eqref{I-a} with
$\mu=1$, \eqref{I-l} with $\mu=\nu=3$, \eqref{I-c} with $\mu=3$,
\eqref{I-k} with $\mu=1$ and $\nu=3$, \eqref{I-l} with $\mu=3$ and
$\nu=1$, and \eqref{I-a} with $\mu=3$, respectively.  Then we have
\begin{eqnarray*}
  &~& \frac{1}{|\Gamma_3|}(a-b)^{N-2}(c+b)^{N-2}(c-b)^{N-2}
[2c \dt f f +(a-c) \b f \dtb f -(a+c)\dt{\b f} \db f]\\
  &=& -|\h{N-2},\dt{\psi}(N-2)|_{_{[2]}} |\h{N-3},\db{\psi}(N-2),
\c E^{3} \psi(N-2)|_{_{[2]}}\\
  &~& +|\h{N-2},\c E^{3} \psi(N-2)|_{_{[2]}} |\h{N-3},\db{\psi}(N-2),
\dt{\psi}(N-2)|_{_{[2]}}\\
  &~& +|\h{N-2},\db{\psi}(N-2)|_{_{[2]}} |\h{N-3},\dt{\psi}(N-2),
\c E^{3} \psi(N-2)|_{_{[2]}}\\
  &=& 0,
\end{eqnarray*}
where we have made use of Lemma \ref{L:lap} in which $\mathbf{B}=(\h{N-3})$,
and $\mathbf{a}=\psi(N-2)$, $\mathbf{b}=\db{\psi}(N-2)$,
$\mathbf{c}=\dt{\psi}(N-2)$ and $\mathbf{d}=\c E^{3} \psi(N-2)$.

The other bilinear equation \eqref{eq:bil-HM-b} can be proved in its
down-hat-shifted version in a similar way by taking $\kappa\equiv 1$.
\end{proof}

\subsubsection{Bilinearization-II}
In fact there is another bilinearization of H3$^{\delta}$ using
\eqref{trans-H3}. Consider the bilinear system
\begin{subequations}
\label{eq:bil-H3-2M}
\begin{eqnarray}
  {\mathcal{B}}_1^{\prime} & \equiv  & (b+c)\th f \b f
 +(a-c)f \th {\b f} -(a+b)\t f\bh f =0,
 \label{eq:bil-H3-2Ma}\\
 {\mathcal{B}}_2^{\prime}& \equiv  & (c-b)\th f \db f
-(a+c)f \th {\db f} +(a+b)\t f\db{\h f} =0,
 \label{eq:bil-H3-2Mb}\\
 {\mathcal{B}}_3^{\prime}& \equiv  & (c-a)(b+c) \tb f
\db{\h f} +(a+c)(b-c) \hb f \db{\t f} +2c(a-b) f \th f =0.
\label{eq:bil-H3-2Mc}
\end{eqnarray}
\end{subequations}
This system is related to H3$^{\delta}$ through
\begin{equation*}
\begin{array}{rl}
  \mathrm{H}3^\delta \equiv & \frac{c}{f\t f\h f \,\th f}
  \biggl [A^2\alpha^{2n}\beta^{2m}\frac{\h f\,\tb f
{\mathcal{B}}_1^{\prime}-\t f\,\bh f {\mathcal{B}}_2^{\prime}}{(a+c)(b+c)}
  +B^2\alpha^{-2n}\beta^{-2m}\frac{\db{\h f}\,\t f
    \db{\mathcal{B}}_1^{\prime}-\db{\t f}\,\h f
\db{\mathcal{B}}_2^{\prime}}{(a-c)(b-c)}\\
  & \quad \quad ~~+AB \Bigl(\frac{\t f\db{\h f} {\mathcal{B}}_2^{\prime}
    +\h f {\tb f} \db {\mathcal{B}}_2^{\prime}}{(a+c)(b-c)}
  -\frac{\h f\db{\t f} {\mathcal{B}}_1^{\prime}
    +\t f {\bh f} \db {\mathcal{B}}_1^{\prime}}{(a-c)(b+c)}
  -\frac{2(a+b)\t f\h f {\mathcal{B}}_3^{\prime}}{(a^2-c^2)(b^2-c^2)}
\Bigr)\biggr].
\end{array}
\end{equation*}

The bilinear system \eqref{eq:bil-H3-2M} shares the same Casoratian
solutions \eqref{eq:bil-HM}:
\begin{prop}
\label{P:H3-Nss-2}
  The Casoratian
  \begin{equation}
  f=|\h{N-1}|_{_{[\nu]}}, \quad  (\nu=1,2~{\rm or}~3),
  \end{equation}
  with entries given by $\psi$ of \eqref{psi-gen}, solves the bilinear
  H3$^{\delta}$ \eqref{eq:bil-H3-2M}.
\end{prop}

\begin{proof}
  We only prove \eqref{eq:bil-H3-2Ma} and \eqref{eq:bil-H3-2Mc}.
  \eqref{eq:bil-H3-2Mb} is similar to \eqref{eq:bil-H3-2Ma}.  We prove
  \eqref{eq:bil-H3-2Ma} in its down-tilde-shifted version, i.e.,
\begin{equation}
(b+c)\h f \dt{\b f} +(a-c)\dt f \bh{f} - (a+b) f \dt{\bh f} =0.
\label{eq:bil-H3-2Mad}
\end{equation}
We need to use Casoratians w.r.t.~bar shift. So we now fix
$\kappa\equiv 3$ in \eqref{Formula-I}.  For $\h f$, $\dt{\b f}$, $\dt
f$, $\bh{f}$ and $\dt{\bh f}$, we use the formulas \eqref{I-d} with
$\mu=2$, \eqref{I-g} with $\mu=1$, \eqref{I-b} with $\mu=1$,
\eqref{I-i} with $\mu=2$ and \eqref{I-n} with $\mu=2$ and $\nu=1$,
respectively, and $f=|\h{N-1}|_{_{[3]}}$ .  Then we have
\begin{eqnarray*}
  &~& \frac{1}{|\Gamma_2|}(a-c)^{N-2}(b+c)^{N-2}[(b+c)\h f \dt{\b f}
 +(a-c)\dt f \bh{f} - (a+b) f \dt{\bh f}]\\
  &=& -|\h{N-2},\c E^{2}{\psi}(N-1)|_{_{[3]}} |\t{N-1},\dt{\psi}(N-1)|_{_{[3]}}\\
  &~& +|\h{N-2},\dt{\psi}(N-1)|_{_{[3]}} |\t{N-1},\c E^{2} \psi(N-1)|_{_{[3]}}\\
  &~& -|\h{N-1}|_{_{[3]}}|\t{N-1},\dt{\psi}(N-1),\c E^{2} \psi(N-1)|_{_{[3]}} \\
  &=& 0,
\end{eqnarray*}
where we have made use of Lemma \ref{L:lap} in which
$\mathbf{B}=(\t{N-2})$, and $\mathbf{a}=\psi(0)$,
$\mathbf{b}=\psi(N-1)$, $\mathbf{c}=\dt{\psi}(N-1)$ and $\mathbf{d}=\c
E^{2} \psi(N-1)$.

For \eqref{eq:bil-H3-2Mc}, after a down-hat shift we get
\begin{equation}
2c (a-b)\dh f \t f -(a-c)(b+c)\dh{\bt f} \db{ f} +(b-c)(a+c)\b f \dhb{\t f} =0.
\label{eq:bil-H3-2Mcd}
\end{equation}
In this case we fix $\kappa\equiv 1$ in \eqref{Formula-I}.  For $\dh
f$, $\t f$, $\dh{\bt f}$, $\db{f}$, $\b f$ and $\dhb{\t f}$, we use
the formulas \eqref{I-b} with $\mu=2$, \eqref{I-n} with $\mu=\nu=3$,
\eqref{I-n} with $\mu=3$ and $\nu=2$, \eqref{I-b} with $\mu=3$,
\eqref{I-d} with $\mu=3$, and \eqref{I-m} with $\mu=2$ and $\nu=3$,
respectively.  Then
\begin{eqnarray*}
  &~& \frac{1}{|\Gamma_3|}(b-a)^{N-2}(c+a)^{N-2}(c-a)^{N-2}
      [2c (a-b)\dh f \t f -(a-c)(b+c)\dh{\bt f} \db{ f}
 +(b-c)(a+c)\b f \dhb{\t f}]\\
  &=& -|\h{N-2},\dh{\psi}(N-1)|_{_{[1]}} |\t{N-2},
\db{\psi}(N-1),\c E^{3}{\psi}(N-1)|_{_{[1]}}\\
  &~& +|\h{N-2},\db{\psi}(N-1)|_{_{[1]}} |\t{N-2},
\dh{\psi}(N-1),\c E^{3}{\psi}(N-1)|_{_{[1]}}\\
  &~& +|\h{N-2},\c E^{3}{\psi}(N-1)|_{_{[1]}}|\t{N-2},
\db{\psi}(N-1),\dh \psi(N-1)|_{_{[1]}}\\
  &=& 0,
\end{eqnarray*}
where in Lemma \ref{L:lap} we take this time $\mathbf{B}=(\t{N-2})$, and
$\mathbf{a}=\psi(0)$, $\mathbf{b}=\db{\psi}(N-1)$,
$\mathbf{c}=\dh{\psi}(N-1)$ and $\mathbf{d}=\c E^{3} \psi(N-1)$.
\end{proof}

\section{Q1 with linear background}
\subsection{Background solution  with $T(x)=x+c$}
Q$1^{\delta}$ is
\begin{equation}
{\rm Q}1^{\delta} \equiv p(u-\h{u})(\t{u}-\th{u})-q(u-\t{u})
(\h{u}-\th{u})-\delta^2pq(q-p)=0.
\label{Q1d}
\end{equation}
With the fixed point defined by $T(x)=x+c$ the side equations for the
background are
\begin{equation}
  \label{Q1-lin-side0}
r(u-\t u)^2=p(c^2+\delta^2r(p-r)),\quad
r(u-\h u)^2=q(c^2+\delta^2r(q-r)).
\end{equation}
After the reparameterization $(p,q)\to(a,b)$ with
\begin{equation}
  \label{Q1-lin-repar}
p=\frac{c^2/r-\delta^2 r}{a^2-\delta^2},\quad
q=\frac{c^2/r-\delta^2 r}{b^2-\delta^2},\quad \alpha:=pa,\quad \beta:=qb,
\end{equation}
equations \eqref{Q1-lin-side0} factorize as in \eqref{H1-0ss-eqs}, and
thus the 0SS will be as in \eqref{A-lin2-0ss}, where, however, we
should replace $a$ and $b$ with $\alpha$ and $\beta$, respectively.

\subsection{1-soliton solution}
The BT for constructing the 1SS is
\begin{subequations}
\label{BT-Q1}
\begin{eqnarray}
p(u-\b{u})(\widetilde{u}-\tb{u})-\varkappa(u-\t{u})(\b{u}-\tb{u})
&=&\delta^2 p\varkappa(\varkappa-p),
\label{BT-Q1-a}\\
\varkappa(u-\h{u})(\b{u}-\hb{u})-q(u-\b{u})(\h{u}- \hb{u})
&=&\delta^2 \varkappa q(q-\varkappa).
\label{BT-Q1-b}
\end{eqnarray}
\end{subequations}
Following the usual procedure we take $u=\alpha n+\beta m+\gamma$ as
the 0SS and
\begin{equation}
\b{u}=\alpha n+\beta m+\gamma+\kappa+v,
\label{ub-Q1-lin}
\end{equation}
as the 1SS, where $v$ is to be determined.  If we now choose
\begin{equation}
\varkappa=\frac{c^2/r-\delta^2 r}{k^2-\delta^2},\quad \kappa=k\varkappa,
\label{vka-Q1}
\end{equation}
then we get \eqref{NM-matricesEG} with
\[
E=-\varkappa(a+k),\quad F=-\varkappa(a-k),\quad
G=-\varkappa(b+k),\quad H=-\varkappa(b-k).
\]
Thus if we define
\[
\rho_{nm}=\left(\frac{a+k}{a-k}\right)^n
\left(\frac{b+k}{b-k}\right)^n\rho_{00}
\]
we find the 1SS in the form
\[
u=\alpha n+\beta m + \gamma+\kappa\frac{1-\rho_{nm}}{1+\rho_{nm}},
\]
where $p,q,\varkappa$ depend on $a,b,k$ as given in
\eqref{Q1-lin-repar},\eqref{vka-Q1}. This is similar to
\eqref{u1ss-H1} except for the more complicated dependence on the
parameters $a,b,k$.

\subsection{NSS}
After studying the 2SS using Hirota's perturbative method we propose
that the NSS is obtained using the following:
\begin{equation}
u_{n,m}^{NSS}  =\alpha n+\beta m + \gamma -(c^2/r-\delta^2 r) \frac{g}{f},
\label{trans-Q1-lin}
\end{equation}
where
\begin{equation}\label{Q1-lin-casdef}
  f=|\h{N-1}|_{_{[3]}},\,  g=|-1,\t{N-1}|_{_{[3]}},
\end{equation}
with $\psi$ defined by
\begin{equation}
\psi_i(n,m,l)= \varrho_{i}^{+}(a +k_i)^n(b +k_i )^m(\delta+ k_i)^l+
  \varrho_{i}^{-}(a -k_i)^n(b -k_i )^m(\delta-k_i)^l,
  \label{psi-Q1-lin}
\end{equation}

The solution \eqref{trans-Q1-lin} is consistent with (I-5.11): Taking
$c=0$ in \eqref{trans-Q1-lin} and $A=D=0, B=\delta/2$ in (I-5.11),
then using (I-2.34c,5.1b,5.4) one finds that
\[
a=\delta,\quad S(-a,a)={\textstyle \sum_j\frac1{k_j-a}}+\frac gf.
\]
However, the additional parameter $c$ \eqref{trans-Q1-lin} implies
more freedom, cf. \eqref{Q1-lin-repar} and (I-5.1b).

The functions $f,g$ satisfy the following bilinear equations (see
proposition below) among others
\begin{subequations}\label{Q1-lin-bileqs}
\begin{eqnarray}
\mathcal Q_1&\equiv&
\bar{\h{\t f}}  f  (b-\delta )+
\h{\t f} \bar{f}  (a+\delta )-\bar{\t{f}} \h{f} (a+b)=0,\label{Q1-lin-bileqs-1}\\
\mathcal Q_2&\equiv &\bar{\h{\t f}} f (a-b)
+\bar{\t{f}} \h{f} (b+\delta )-\t{f} \bar{\h{f}} (a+\delta )=0,\label{Q1-lin-bileqs-2}\\
\mathcal Q_3&\equiv&-\bar{\t{f}} \h{f}+\bar{\t{f}} \h{g} (-a+\delta )
+\t{f} \bar{\h{f}}+\bar{\h{f}} \t{g} (b-\delta)
+\bar{f} \h{\t g} (a-b)=0,\label{Q1-lin-bileqs-3}\\
\mathcal Q_4&=&\bar{\h{\t f}} g (a-b)+\bar{\t{f }} \h{g} (a+b)
-\bar{\h{f }}\t{g} (a+b)+\bar{f} \h{\t g} (-a+b)=0.\label{Q1-lin-bileqs-4}
\end{eqnarray}
\end{subequations}

We note that $\mathcal Q_4$ can also be replaced by
\begin{align}
  \mathcal Q'_4= \mathcal Q_3+\mathcal Q_4 = \hb f\t f -\hb f \t g
  (a+\delta) -\bt f \h f +\bt f \h g (\delta+b)+\th{\b f} g(a-b) =0,
\end{align}
which is similar to $\mathcal Q_3$.  When the dependent variable
transformation \eqref{trans-Q1-lin} is substituted into Q1$^\delta$ we
find that the result can be expressed in terms of the $\mathcal Q_i$
defined above:
\begin{equation}
\mathrm{Q}1^\delta=\frac{(c^2/r-\delta^2r)^3}
{(a^2-\delta^2)(b^2-\delta^2)(a-b)(a+\delta)\bar{f}
f \t f \h f \h{\t f}}\quad
\sum^{4}_{i=1}\mathcal{Q}_iP_i,
\end{equation}
with
\begin{subequations}
\begin{eqnarray}
P_1&=&(a-b) [\t{f} \h{f} g (-a+b)
+f (\h{f} \t{g}-\t{f} \h{g}) (a+b)\nonumber\\
&&\phantom{(a-b) [}
+\t{f} \h{g} g (-a^2+\delta ^2)
+\h{f} \t{g} g (b^2-\delta ^2)
+f \t{g} \h{g} (a^2-b^2)],\\
P_2&=&(a+b) [
\t{f} \h{f} g (a-b)
+f (\t{f} \h{g}-\h{f} \t{g}) (b-\delta )
+\t{g} (\h{f} g-f \h{g}) (a-b) (b-\delta )],~\\
P_3&=&(a+b) (a+\delta ) [
 \t{f} \h{f} g (a-b)
+\t{f} f \h{g} (b-\delta )
+\h{f} f \t{g} (-a+\delta )],\\
P_4&=&f (a+\delta ) [
\t{f} \h{f} (-a+b)
+(\t{f} \h{g}-\h{f} \t{g}) (a-\delta ) (b-\delta )].
\end{eqnarray}
\end{subequations}
Then it remains to prove the following:
\begin{prop}
\label{P:Q1-lin-Nss}
The Casoratian type determinants $f$, $g$, given in
\eqref{Q1-lin-casdef} with entries given by \eqref{psi-Q1-lin} solve
the set of bilinear equations \eqref{Q1-lin-bileqs}.
\end{prop}

\begin{lem}
\label{L:equl}
By means of Lemma \ref{L:lap}, the following formulas are zero, ($\mu=1,2$),
\begin{subequations}
\begin{align}
Y_{\mu}=&f|E_{\mu}{\psi}(-1),-1,\t{N-2}|_{_{[3]}}+\db{f}|E_{\mu}{\psi}(-1),\t{N-1}|_{_{[3]}}
-g|E_{\mu}{\psi}(-1),\h{N-2}|_{_{[3]}},
\label{Ymu}\\
Y_3=&|\dt{\psi}(-1),-1,\t{N-2}|_{_{[3]}}|\dh{\psi}(-1),\t{N-1}|_{_{[3]}}
-|\dh{\psi}(-1),-1,\t{N-2}|_{_{[3]}}|\dt{\psi}(-1),\t{N-1}|_{_{[3]}}\nonumber\\
&+g|\dh{\psi}(-1),\dt{\psi}(-1),\t{N-2}|_{_{[3]}};
\label{Y3}
\end{align}
\label{Y}
\end{subequations}
\begin{subequations}
\begin{align}
Z_{\mu}=&f|\c E^{\mu}{\psi}(-1),-1,\t{N-2}|_{_{[3]}}+\db{f}|\c E^{\mu}{\psi}(-1),\t{N-1}|_{_{[3]}}
-g|\c E^{\mu}{\psi}(-1),\h{N-2}|_{_{[3]}},
\label{Zmu}\\
Z_3=&|\c E^{1}{\psi}(-1),-1,\t{N-2}|_{_{[3]}}|\c E^{2}{\psi}(-1),\t{N-1}|_{_{[3]}}\nonumber\\
&-|\c E^{2}{\psi}(-1),-1,\t{N-2}|_{_{[3]}}|\c E^{1}{\psi}(-1),\t{N-1}|_{_{[3]}}\nonumber\\
&+g|\c E^{2}{\psi}(-1),\c E^{1}{\psi}(-1),\t{N-2}|_{_{[3]}}.
\label{Z3}
\end{align}
\label{Z}
\end{subequations}
\end{lem}
\begin{proof}
For the above formulas, we can  use Lemma \ref{L:lap} by respectively taking
\begin{align*}
\mathbf{B}=(\t{N-2}),~~&
(\mathbf{a}, \mathbf{b}, \mathbf{c}, \mathbf{d})
=(E_{\mu}{\psi}(-1),\psi(-1),\psi(0),\psi(N-1)),\\
\mathbf{B}=(\t{N-2}),~~&
(\mathbf{a}, \mathbf{b}, \mathbf{c}, \mathbf{d})
=(\dh{\psi}(-1),\dt \psi(-1),\psi(-1),\psi(N-1));\\
\mathbf{B}=(\t{N-2}),~~&
(\mathbf{a}, \mathbf{b}, \mathbf{c}, \mathbf{d})
=(\c E^{\mu}{\psi}(-1),\psi(-1),\psi(0),\psi(N-1)),\\
\mathbf{B}=(\t{N-2}),~~&
(\mathbf{a}, \mathbf{b}, \mathbf{c}, \mathbf{d})
=(\c E^2{\psi}(-1),\c E^1 \psi(-1),\psi(-1),\psi(N-1)).
\end{align*}
\end{proof}

\begin{proof}
\

\noindent\textit{Proof for \eqref{Q1-lin-bileqs-1}:}~~
The down-bar shifted \eqref{Q1-lin-bileqs-1} is nothing but
$\mathcal{B}'_2$ of H3$^\delta$ with $c=\delta$.

\noindent\textit{Proof for \eqref{Q1-lin-bileqs-2}:}~~
By a down-tilde shift \eqref{Q1-lin-bileqs-2}
is written as
\begin{equation}
(a-b)\dt f \hb f +(\delta+b)\b f \h{\dt f}-(a+\delta)\bh{\dt f}f =0.
\label{Q1-lin-bileqs-21}
\end{equation}
We fix $\kappa\equiv 2$ and $c=\delta$ in \eqref{Formula-I}.  For $\dt
f$, $\bh{f}$, $\b f$, $\h{\dt{f}}$ and $\bh{\dt f}$, we use the
formulas \eqref{I-b} with $\mu=1$, \eqref{I-i} with $\mu=3$,
\eqref{I-d} with $\mu=3$, \eqref{I-g} with $\mu=3$ and \eqref{I-n}
with $\mu=3$ and $\nu=1$, respectively, and $f=|\h{N-1}|_{_{[2]}}$.
Then we have
\begin{eqnarray*}
  &~& \frac{1}{|\Gamma_3|}(a-b)^{N-2}(\delta+b)^{N-2}
  [(a-b)\dt f \hb f +(\delta+b)\b f \h{\dt f}-(a+\delta)\bh{\dt f}f]\\
  &=& |\h{N-2},\dt{\psi}(N-1)|_{_{[2]}} |\t{N-1},\c E^{3}{\psi}(N-1)|_{_{[2]}}\\
  &~& -|\h{N-2},\c E^{3}{\psi}(N-1)|_{_{[2]}} |\t{N-1}, \dt \psi(N-1)|_{_{[2]}}\\
  &~& -|\h{N-1}|_{_{[2]}}|\t{N-1},\dt{\psi}(N-1),\c E^{3} \psi(N-1)|_{_{[2]}} \\
  &=& 0,
\end{eqnarray*}
where we have made use of Lemma \ref{L:lap} with
\begin{equation*}
\mathbf{B}=(\t{N-2}),~~
(\mathbf{a},\mathbf{b},\mathbf{c},\mathbf{d})=
(\psi(0),\psi(N-1), \dt{\psi}(N-1), \c E^{3} \psi(N-1)).
\end{equation*}

\noindent\textit{Proof for \eqref{Q1-lin-bileqs-3}:}~~
By a down-tilde-hat shift \eqref{Q1-lin-bileqs-3}
is written as
\begin{equation}
  -\dh{\b f}[\dt f +(a-\delta)\dt g]+ \dt{\b f}[\dh f +
(b-\delta)\dh g]+(a-b)\dth{\b f}g =0.
\label{Q1-lin-bileqs-31}
\end{equation}
This time we fix $\kappa\equiv 3$ and $c=\delta$ in
\eqref{Formula-II}.  For $\dh{\b f}$, $\dt f$, $\dt g$, $\dt{\b f}$,
$\dh f$, $\dh g$ and $\dth{\b f}$, we use the formulas \eqref{II-b}
with $\mu=2$, \eqref{II-a} with $\mu=1$, \eqref{II-c} with $\mu=1$,
\eqref{II-b} with $\mu=1$, \eqref{II-a} with $\mu=2$, \eqref{II-c}
with $\mu=2$, and \eqref{II-g}, respectively.  Then we have
\begin{eqnarray*}
  &~& -\dh{\b f}[\dt f +(a-\delta)\dt g]+ \dt{\b f}[\dh f
 +(b-\delta)\dh g]+(a-b)\dth{\b f}g \\
  &=& (a-\delta)^2 Y_1 -(b-\delta)^2 Y_2+(a-\delta)^2(b-\delta)^2 Y_3
\end{eqnarray*}
with $Y_j$ defined in \eqref{Y}, which are zero in the light of Lemma
\ref{L:equl}.

\noindent\textit{Proof for \eqref{Q1-lin-bileqs-4}:}~~
Using \eqref{Q1-lin-bileqs-3} to eliminate the term $(a-b)\b f \th g$
from \eqref{Q1-lin-bileqs-4}, we get
\begin{equation}
\bh f[\t f -(a+\delta)\t g]- \bt f[\h f -(b+\delta)\h g]+(a-b)\th{\b f}g =0.
\label{Q1-lin-bileqs-41}
\end{equation}
We fix $\kappa\equiv 3$ and $c=\delta$ in \eqref{Formula-II},
and use the formulas \eqref{II-d}, \eqref{II-e}, \eqref{II-f} and \eqref{II-h}.
Then it turns out that
\begin{eqnarray*}
  &~& \bh f[\t f -(a+\delta)\t g]- \bt f[\h f -(b+\delta)\h g]+(a-b)\th{\b f}g \\
  &=& (a+\delta)^2 Z_1 -(b+\delta)^2 Z_2 -(a+\delta)^2(b+\delta)^2 Z_3
\end{eqnarray*}
with $Z_j$ defined in \eqref{Z}, which are zero in the light of Lemma
\ref{L:equl}.

Thus we have completed the proof for all bilinear equations.

\end{proof}

\section{Q1 with power background}
\subsection{Background solution with $T(x)=-x+c$}
With fixed point defined by $T(x)=-x+c$ we use the reparameterization
\begin{align}
  \label{Q1m-repar}
p&=\tfrac12 r(1- \cosh(\alpha'))=-\tfrac14r(1-\alpha)^2/\alpha,\\
q&=\tfrac12 r(1- \cosh(\beta'))=-\tfrac14r(1-\beta)^2/\beta,\\
 x_{nm}&=A e^{y_{nm}}+B e^{y_{nm}}+c/2,\quad\text{ where} \quad AB=\delta^2r^2/16,
\end{align}
and this leads to equations that factorize as in \eqref{H3-exp-fact}.
Thus we get power type background solutions for $u$:
\begin{equation}
  \label{Q1-pow-0ss}
  u=\tfrac12c+A \alpha^n \beta^m +B \alpha^{-n}\beta^{-m},\quad
AB=\delta^2 r^2/16.
\end{equation}
Here $c$ is related to translation freedom and $r$ to scaling freedom;
in the following we take $c=0$.

\subsection{1-soliton solution for Q$1^{\delta}$}
In order to derive the 1SS we use the BT \eqref{BT-Q1}
with the seed solution \eqref{Q1-pow-0ss} with
\begin{equation}
\b{u}=\b{u}_{0}+v,
\label{u-bar-Q1}
\end{equation}
where $\b{u}_{0}$ is the bar-shifted background solution
\eqref{Q1-pow-0ss}:
\begin{equation}
\b{u}_{0}=A \alpha^n \beta^m \kappa +B
\alpha^{-n}\beta^{-m}\kappa^{-1},
\quad AB=\delta^2r^2/16.
\label{u-b-theta-Q1d}
\end{equation}
and $\kappa$ is defined through
\begin{equation}
\varkappa=-\tfrac r4 (1-\kappa)^2/\kappa.
\label{Q1-K}
\end{equation}

If we now define
\begin{equation}
U_{n,m}=A \alpha^n \beta^m -B \alpha^{-n}\beta^{-m}
\label{U-Q1-1}
\end{equation}
then it is straightforward to derive \eqref{NM-matricesST} from
\eqref{BT-Q1} with
\begin{equation}
S=\frac{r(1-\alpha)(1-\kappa)(1-\alpha \kappa)}{4\alpha \kappa},\quad
\Delta=\frac{r(1-\alpha)(1-\kappa)(\alpha -\kappa)}{4\alpha
  \kappa},\quad \sigma=p,
\end{equation}
\begin{equation}
T=\frac{r(1-\beta)(1-\kappa)(1-\beta \kappa)}{4\beta \kappa},\quad
\Omega=\frac{r(1-\beta)(1-\kappa)(\beta -\kappa)}{4\beta \kappa},\quad
\tau=q.
\end{equation}
On the basis of this we define $\rho$ as usual by
\begin{equation}
\rho_{n,m}=\biggl(\frac{S}{\Delta}\biggr)^n\biggl(
\frac{T}{\Omega}\biggr)^m \rho_{0,0}
=\biggl(\frac{1-\alpha \kappa}{\alpha-\kappa}\biggr)^n\biggl(
\frac{1-\beta \kappa}{\beta-\kappa}\biggr)^m \rho_{0,0},
\label{rho-Q1}
\end{equation}
where $\rho_{0,0}$ is some constant. Then it follows that
\begin{equation}
v_{n,m}=\frac{\frac{1-\kappa^2}{\kappa} U_{n,m} \rho_{n,m}}
{1+\rho_{n,m}},
\end{equation}
and finally we obtain the 1-soliton for Q$1^{\delta}$:
\begin{eqnarray}
u_{n,m}^{1SS}&=&\b{u}_{0}+v_{n,m}\\
&=& \frac{A \alpha^n \beta^m(\kappa+\kappa^{-1}\rho_{n,m})
+B \alpha^{-n}\beta^{-m}(\kappa^{-1}+ \kappa\rho_{n,m})}{1+\rho_{n,m}}\\
&=& \frac{A' \alpha^n \beta^m(1+\kappa^{-2}\rho_{n,m})
+B' \alpha^{-n}\beta^{-m}(1+ \kappa^2\rho_{n,m})}{1+\rho_{n,m}},
\label{1ss-Q1}
\end{eqnarray}
where $A'B'=AB=\delta^2r^2/16$.

\subsection{Bilinearization}
In order to get $\rho$ of \eqref{rho-Q1} into a nicer form we use the
M\"obius transformations
\begin{equation}
\alpha =\frac{a-c}{a+c}, \quad
\beta =\frac{b-c}{b+c}, \quad
\kappa=\frac{k-c}{k+c},\quad\Rightarrow
p=\frac{rc^2}{a^2-c^2},\quad q=\frac{rc^2}{b^2-c^2},
\label{Mob-trans-Q1}
\end{equation}
which leads to the canonical form
\begin{equation}
\rho_{n,m}=\biggl(\frac{a+k}{a-k}\biggr)^n
\biggl(\frac{b+k}{b-k}\biggr)^m \rho_{0,0}.
\label{rho-Q1M}
\end{equation}
With the above M\"obius transformations Q1$^\delta$ \eqref{Q1d} can be
written as
\begin{equation}
{\rm Q}1^{\delta} \equiv (b^2-c^2)(u-\h{u})(\t{u}-\th{u})
-(a^2-c^2)(u-\t{u})(\h{u}-\th{u})
- \frac{r^2\delta^2 c^4 (a^2-b^2)}{(a^2-c^2)(b^2-c^2)}=0,
\label{Q1dm}
\end{equation}
Using \eqref{rho-Q1M} and rearranging the parameters $A$ and $B$ we
write the 1SS \eqref{1ss-Q1} as
\begin{equation}
u_{n,m}^{1SS} = A \alpha^n \beta^m \frac{\psi(n,m,l+2)}{\psi(n,m,l)}
+B \alpha^{-n}\beta^{-m}\frac{\psi(n,m,l-2)}{\psi(n,m,l)},
\quad AB=\delta^2r^2/16,
\label{1ss-Q1-M}
\end{equation}
where $\psi$ is as in \eqref{psi-gen}.

This structure motivates us to bilinearize Q1$^\delta$ through the
following transformation,
\begin{equation}
  u_{n,m}^{NSS}  = A \alpha^n \beta^m\frac{\bb f}{f}+
  B \alpha^{-n}\beta^{-m}\,\frac{\dbb f}{f},\quad AB=\delta^2r^2/16,
\label{trans-Q1}
\end{equation}
where $f=|\h{N-1}|_{_{[\nu]}}~(\nu=1,2,~\mathrm{or}~3)$ with entries
\eqref{psi-gen} and the bar-shift is the shift in the third index $l$,
as discussed in Sec. \ref{sec:cas}.

The solution \eqref{trans-Q1} with $r=1$ is the same as
(I-5.11) with $B=0$.  In fact, by a comparison \eqref{Mob-trans-Q1}
with (I-3.17) one has
\[
\rho(a)=\alpha^{-n}\beta^{-m},~~ \rho(-a)=\alpha^{n}\beta^{m};
\]
and substituting $\rho_i$ in (I-2.2) by \eqref{rho-I} we find from (I-3.17) that
\[
(1+2a S(-a,-a))\times ({\textstyle \prod (k_j-a)^2})=\frac{\bb f}{f},
\quad (1-2aS(a,a))\,/\,({\textstyle \prod (k_j-a)^2})=\frac{\dbb f}{f}.
\]

Surprisingly Q1$^\delta$ can also be bilinearized with the same
bilinear equations as H3$^\delta$, namely \eqref{eq:bil-HM}, and we
have already shown in Proposition \ref{P:H3-Nss-1} that
$f=|\h{N-1}|_{_{[\nu]}}$ solves the $\mathcal B_i$ in \eqref{eq:bil-HM}.

The representation of Q1$^\delta$ \eqref{Q1dm} in terms of $\mathcal
B_i$ is as follows:
\begin{equation*}
\mathrm{Q}1^\delta \equiv \frac{\alpha^{4n+2}\beta^{4m+2}
(a+c)^2(b+c)^2 A^2 \b P_1
-\alpha^{2n}\beta^{2m} \delta^2/4 P_2-(a+c)^2(b+c)^2 B^2\db P_1}
{ \alpha^{2n+1}\beta^{2m+1}(a+c)^2(b+c)^2 f\t f\h f \th f},
\end{equation*}
where
\begin{equation*}
\begin{array}{rl}
  P_1=& Y\t Y-X\h X,~~ X=\mathcal{B}_1-2cf\t f,~~ Y=\mathcal{B}_2-2cf\h f,\\
  P_2=& -(a+c)(a-c)(b+c)^2 \Bigl(\b X\db{\h X}- 4c^2 \b f\,\tb f \db {\h f}\db {\th f} \,\Bigr)\\
      & -(a+c)(a-c)(b-c)^2 \Bigl(\db X \bh X -4c^2 \db f\,\db{\t f}\,\hb f\,\th {\b f}\,\Bigr)\\
      & +4c^2(b+c)(b-c) \Bigl(X \h X -4c^2 f\t f \,\h f\, \th f \, \Bigr)\\
      & +(b+c)(b-c)(a+c)^2 \Bigl(\b Y \db{\t Y}-4c^2 \b f\,\bh f \,\db {\t f}\,\db {\th f}\, \Bigr)\\
      & +(b+c)(b-c)(a-c)^2 \Bigl(\db Y \bt Y - 4c^2 \db f\db{\h f}\,\tb f\,\th {\b f}\, \Bigr)\\
      &  -4c^2(a+c)(a-c) \Bigl(Y \t Y-4c^2 f\h f \,\t f\, \th f\, \Bigr).
\end{array}
\end{equation*}

\section{Conclusions}
In this paper, companion to \cite{NAJ} we have analyzed the soliton
solutions to the models H1, H2, H3, and Q1 in the ABS list
\cite{ABS-CMP-2002} of partial difference equations. Our method is
constructive, progressing in each case from background solution to 1SS
to NSS and bilinearization.

Our approach is fairly algorithmic, and as such we hope that it will
be usable also for other models with multidimensional consistency.
One interesting feature is that in some cases we need several bilinear
equations, some of which seem to have the same continuum limit.
This is a remainder that there are several ways to discretize a
derivative.

\section*{Acknowledgments}
This project is supported by the National Natural Science Foundation
of China (10671121), and Shanghai Leading Academic Discipline
Project (No.J50101).

\begin{appendix}

\section{Casoratians and formulas}
\label{A:a}

\noindent
We define Casoratians
$$f_{_{[\nu]}}=|\h{N-1}|_{_{[\nu]}},\quad  g_{_{[\nu]}}=|\h{N-2},N|_{_{[\nu]}}, \quad
h_{_{[\nu]}}=|\h{N-2},N+1|_{_{[\nu]}},$$
where the matrix entries are given by the function
\begin{eqnarray}
  \psi_i(n,m,l)=
  \varrho_{i}^{+}(c+ k_i)^l(a +k_i)^n(b +k_i )^m+
  \varrho_i^{-}(c- k_i)^l(a -k_i)^n(b -k_i)^m.
  \label{entry-H3-1}
\end{eqnarray}
We introduce notations
\begin{equation}
\Gamma_{\nu,i}=\alpha^2_{\nu} -k_i^2,\quad
\Gamma_{\nu}={\rm Diag}(\Gamma_{\nu,1},\Gamma_{\nu,2},\cdots,\Gamma_{\nu,N}),\quad (\nu=1,2,3),
\end{equation}
where
\begin{equation}
\alpha_1\equiv a,\quad \alpha_2 \equiv b,\quad \alpha_3\equiv c,
\end{equation}
and define the operator $\c E^{\nu}$ as
\begin{equation}
\c E^{\nu}\psi =\Gamma_{\nu}^{-1}E^{\nu} \psi, \quad (\nu=1,2,3).
\label{E-circ}
\end{equation}

The column vector $\psi$ satisfies
\begin{subequations}
\label{rela-I}
\begin{eqnarray}
  (\alpha_\mu-\alpha_\kappa)\psi &\!\!\!=\!\!\!&(E^{\mu}-E^{\kappa})\psi,\label{rela-I-a}\\
  \hskip -25pt(\alpha_\mu\!\!+\alpha_\nu)\c E^{\mu}\!E_{\nu}\psi &\!\!\!=\!\!\!&(\c E^{\mu}+E_{\nu})\psi,
  \label{rela-I-b}\\
  E_{\mu}\psi &\!\!\!=\!\!\!&[E_{\kappa}-(\alpha_\mu-\alpha_\kappa)(E_{\kappa})^{2}
                 +(\alpha_\mu-\alpha_\kappa)^2 E_{\mu}(E_{\kappa})^{2}]\psi,\label{rela-I-c}\\
  E_{\mu}\psi &\!\!\!=\!\!\!\!&[E_{\kappa}-(\alpha_\mu\!-\alpha_\kappa)(E_{\kappa})^{2}
                 +(\alpha_\mu-\alpha_\kappa)^2 (E_{\kappa})^{3}\!
                 -(\alpha_\mu-\alpha_\kappa)^3 E_{\mu}(E_{\kappa})^{3}]\psi.\label{rela-I-d}
\end{eqnarray}
\end{subequations}
For the Casoratian $f_{_{[\kappa]}}$, we have
\begin{subequations}
\label{Formula-I}
\begin{eqnarray}
  -(\alpha_\mu-\alpha_\kappa)^{N-2} E_{\mu} f_{_{[\kappa]}}
  \!\!&=& \!\!|\h{N-2},E_{\mu} \psi(N-2)|_{_{[\kappa]}},\label{I-a}\\
  (\alpha_\mu-\alpha_\kappa)^{N-1} E_{\mu} f_{_{[\kappa]}}
  \!\!&=&\!\! |\h{N-2},E_{\mu} \psi(N-1)|_{_{[\kappa]}},\label{I-b}\\
  (\alpha_\mu+\alpha_\kappa)^{N-2} E^{\mu} f_{_{[\kappa]}}
  \!\!&=& \!\!|\Gamma_{\mu}||\h{N-2},\c E^{\mu} \psi(N-2)|_{_{[\kappa]}},~(\mu\neq \kappa),\label{I-c}\\
  (\alpha_\mu+\alpha_\kappa)^{N-1} E^{\mu} f_{_{[\kappa]}}
  \!\!&=&\!\! |\Gamma_{\mu}||\h{N-2},\c E^{\mu}\psi(N-1)|_{_{[\kappa]}},~(\mu\neq \kappa);\label{I-d}
\end{eqnarray}
\begin{eqnarray}
  E^{\kappa} f_{_{[\kappa]}} \!\!&=& \!\!|\t{N}|_{_{[\kappa]}},\label{I-e}\\
  E_{\kappa} f_{_{[\kappa]}} \!\!&=& \!\!|-1,\h{N-2}|_{_{[\kappa]}},\label{I-f}\\
  -(\alpha_\mu-\alpha_\kappa)^{N-2} E_{\mu}E^{\kappa} f_{_{[\kappa]}}
  \!\!&=& \!\!|\t{N-1},E_{\mu} \psi(N-1)|_{_{[\kappa]}},\label{I-g}\\
  (\alpha_\mu-\alpha_\kappa)^{N-1} E_{\mu}E_{\kappa} f_{_{[\kappa]}}
  \!\!&=& \!\!|-1,\h{N-3},E_{\mu} \psi(N-2)|_{_{[\kappa]}},\label{I-h}\\
  (\alpha_\mu+\alpha_\kappa)^{N-2} E^{\mu}E^{\kappa} f_{_{[\kappa]}}
  \!\!&=& \!\!|\Gamma_{\mu}||\t{N-1},\c E^{\mu}  \psi(N-1)|_{_{[\kappa]}},~(\mu\neq \kappa),\label{I-i}\\
  (\alpha_\mu+\alpha_\kappa)^{N-1} E^{\mu}E_{\kappa} f_{_{[\kappa]}}
  \!\!&=& \!\!|\Gamma_{\mu}||-1,\h{N-3},\c E^{\mu}  \psi(N-2)|_{_{[\kappa]}},~(\mu\neq \kappa);\label{I-j}
\end{eqnarray}
\begin{align}
  (\alpha_\mu-\alpha_\nu)(\alpha_\mu-\alpha_\kappa)^{N-2}
  (\alpha_\nu-\alpha_\kappa)^{N-2}&
  E_{\mu}E_{\nu}f_{_{[\kappa]}}\nonumber\\
  =& |\h{N-3},E_{\nu} \psi(N-2),E_{\mu} \psi(N-2)|_{_{[\kappa]}},\label{I-k}\\
  (\alpha_\mu+\alpha_\nu)(\alpha_\mu+\alpha_\kappa)^{N-2}
  (\alpha_\nu-\alpha_\kappa)^{N-2}&
  E^{\mu}E_{\nu}f_{_{[\kappa]}}\nonumber\\
  =& |\Gamma_{\mu}||\h{N-3},E_{\nu} \psi(N-2),\c E^{\mu} \psi(N-2)|_{_{[\kappa]}},\quad (\mu\neq \kappa);\label{I-l}
\end{align}
\begin{align}
  (\alpha_\mu-\alpha_\nu)(\alpha_\mu-\alpha_\kappa)^{N-2}
  (\alpha_\nu-\alpha_\kappa)^{N-2}&
  E_{\mu}E_{\nu}E^{\kappa} f_{_{[\kappa]}}\nonumber\\
  =& |\t{N-2},E_{\nu} \psi(N-1),E_{\mu} \psi(N-1)|_{_{[\kappa]}},\label{I-m}\\
  (\alpha_\mu+\alpha_\nu)(\alpha_\mu+\alpha_\kappa)^{N-2}
  (\alpha_\nu-\alpha_\kappa)^{N-2}&
  E^{\mu}E_{\nu}E^{\kappa} f_{_{[\kappa]}}\nonumber\\
  =&  |\Gamma_{\mu}| |\t{N-2},E_{\nu} \psi(N-1),\c E^{\mu}  \psi(N-1)|_{_{[\kappa]}},
  ~~(\mu\neq \kappa).\label{I-n}
\end{align}
\end{subequations}
For $g_{_{[\kappa]}}$ we have
\begin{subequations}
\label{Formula-I-g}
\begin{eqnarray}
  \hskip -50pt
  -(\alpha_\mu\!-\!\alpha_\kappa)^{N-2} E_{\mu} [g_{_{[\kappa]}}\!+(\alpha_\mu\!-\!\alpha_\kappa)f_{_{[\kappa]}}]
  \!\!&=&\!\!|\h{N-3},N-1,E_{\mu} \psi(N-2)|_{_{[\kappa]}},\label{I-g-a}\\
  \hskip -20pt(\alpha_\mu\!+\!\alpha_\kappa)^{N-2} E^{\mu}
  [g_{_{[\kappa]}}\!-(\alpha_\mu\!+\alpha_\kappa)f_{_{[\kappa]}}]
  \!\!&=& \!\!|\Gamma_{\mu}||\h{N-3},N\!-1,\c E^{\mu} \psi(N-2)|_{_{[\kappa]}},~(\mu\neq \kappa).
  \label{I-g-b}
\end{eqnarray}
\end{subequations}
For $h_{_{[\kappa]}}$,
\begin{subequations}
\label{Formula-I-h}
\begin{eqnarray}
  \hskip -40pt
  -(\alpha_\mu-\alpha_\kappa)^{N-2} E_{\mu} [h_{_{[\kappa]}}+(\alpha_\mu-\alpha_\kappa)g_{_{[\kappa]}}]
  \!\!&=&\!\!|\h{N-3},N,E_{\mu} \psi(N-2)|_{_{[\kappa]}},\label{I-h-a}\\
  \hskip -20pt(\alpha_\mu+\alpha_\kappa)^{N-2} E^{\mu} [h_{_{[\kappa]}}-(\alpha_\mu+\alpha_\kappa)g_{_{[\kappa]}}]
  \!\!&=&\!\!|\Gamma_{\mu}||\h{N-3},N,\c E^{\mu}\psi(N-2)|_{_{[\kappa]}},~~(\mu\neq \kappa).
  \label{I-h-b}
\end{eqnarray}
\end{subequations}
These formulas, \eqref{rela-I}-\eqref{Formula-I-h}, are valid for
$\mu,\nu,\kappa=1,2,3$ except when otherwise indicated.

\section{Proof for the formulas in Appendix \ref{A:a}}

The formulas \eqref{rela-I} can directly be verified.  In
\eqref{Formula-I}-\eqref{Formula-I-h}, those only containing down
shifts $E_{\mu}$ and $E_{\nu}$ can be proved by using the relations
\eqref{rela-I}.  Since $a$, $b$ and $c$ appear in $\psi_i$
equivalently, we take $\mu=1$, $\nu=2$ and $\kappa=3$ in the following
proof.

From \eqref{rela-I-a} we have
\begin{eqnarray}
(a-c)  \dt \psi(l)&=&\psi(l)-\dt\psi(l+1),
\label{rela-I-a-1}\\
(b-c)  \dh \psi(l)&=&\psi(l)-\dh\psi(l+1),
\label{rela-I-a-2}\\
(a-b)  \dth \psi(l)&=&\dh \psi(l)-\dt\psi(l).
\label{rela-I-a-3}
\end{eqnarray}
Using \eqref{rela-I-a-1} for $f_{_{[3]}}$  we first find
\[
(a-c) \dt f_{_{[3]}} =|(a-c)\dt \psi(0),\dt \psi(1),\dots,\dt \psi(N-1)|_{_{[3]}}
=| \psi(0),\dt \psi(1),\dots,\dt \psi(N-1)|_{_{[3]}}
\]
and repeating this $N-1$ times we get
\[
  (a-c)^{N-1} \dt f_{_{[3]}}=| \psi(0),\dots,\psi(N-2),\dt \psi(N-1)|_{_{[3]}}
=|\h{N-2},\dt \psi(N-1)|_{_{[3]}}.
\]
This is formula \eqref{I-b}. Again using \eqref{rela-I-a-1} and
expressing $\dt \psi(N-1)$ through $\psi(N-2)$ and $\dt\psi(N-2)$ in
the above formula yields
\begin{equation}
  -(a-c)^{N-2} \dt f_{_{[3]}}=|\h{N-2},\dt \psi(N-2)|_{_{[3]}},
\end{equation}
which is \eqref{I-a}.  For $\dth f_{_{[3]}}$ we first using
\eqref{rela-I-a-2} from the above formula to obtain
\[
-(b-c)^{N-2}(a-c)^{N-2} \dth f_{_{[3]}}=|\h{N-3},\dh \psi(N-2),\dth
\psi(N-2)|_{_{[3]}}
\]
then from \eqref{rela-I-a-3} we get
\begin{equation}
  (a-b)(a-c)^{N-2}(b-c)^{N-2} \dth f
=|\h{N-3},\dh \psi(N-2),\dt \psi(N-2)|_{_{[3]}},
\end{equation}
which is \eqref{I-k}.

For $g_{_{[3]}}=|\h{N-2},N|_{_{[3]}}$  we find first
\[
  (a-c)^{N-2} \dt g_{_{[3]}}=|\h{N-3},\dt \psi(N-2),\dt \psi(N)|_{_{[3]}}.
\]
Using \eqref{rela-I-c} to re-write the last column we get
\begin{equation}
  (a-c)^{N-2} \dt g_{_{[3]}}=-|\h{N-3},N-1,\dt \psi(N-2)|_{_{[3]}}
  +(a-c) |\h{N-2},\dt \psi(N-2)|_{_{[3]}},
\end{equation}
which can also be stated as
\begin{equation}
  -(a-c)^{N-2} [\dt g_{_{[3]}}+(a-c)\dt f_{_{[3]}}]=|\h{N-3},N-1,\dt \psi(N-2)|_{_{[3]}},
\end{equation}
i.e., \eqref{I-g-a}. Formula \eqref{I-h-a} can be derived similarly by
using \eqref{rela-I-d}.

To show how to derive those formulas with up shifts in
\eqref{Formula-I}-\eqref{Formula-I-h}, we introduce two auxiliary
functions:
\begin{subequations}
\begin{eqnarray}
  \hskip -20pt \mathbf{II}:& \phi_i(n,m,l)=
  \varrho_{i}^{+}(c+ k_i)^l(a +k_i)^n(b -k_i )^{-m}+
  \varrho_i^{-}(c- k_i)^l(a -k_i)^n(b +k_i)^{-m},
\label{entry-H3-2}\\
  \hskip -40 pt \mathbf{III}:& \omega_i(n,m,l)=
  \varrho_{i}^{+}(c+ k_i)^l(a -k_i)^{-n}(b +k_i )^m+
  \varrho_i^{-}(c- k_i)^l(a +k_i)^{-n}(b -k_i)^m.
\label{entry-H3-3}
\end{eqnarray}
\end{subequations}
Casoratians $f$, $g$ and $h$ (w.r.t.~bar-shift) with column vectors
$\phi=(\phi_1,\cdots,\phi_N)^T$ and
$\omega=(\omega_1,\cdots,\omega_N)^T$ are denoted by $f_{_{{\rm
      II}[3]}}$, $g_{_{{\rm II}[3]}}$, $h_{_{{\rm II}[3]}}$ and
$f_{_{{\rm III}[3]}}$, $g_{_{{\rm III}[3]}}$, $h_{_{{\rm III}[3]}}$,
respectively.  They are related to $f_{_{[3]}}$, $g_{_{[3]}}$ and
$h_{_{[3]}}$ through
\begin{subequations}
\begin{eqnarray}
&& f_{_{[3]}}= |\Gamma_2|^{m} f_{_{{\rm II}[3]}}=|\Gamma_1|^{n} f_{_{{\rm III}[3]}},\\
&& g_{_{[3]}}= |\Gamma_2|^{m} g_{_{{\rm II}[3]}}=|\Gamma_1|^{n} g_{_{{\rm III}[3]}},\\
&& h_{_{[3]}}= |\Gamma_2|^{m} h_{_{{\rm II}[3]}}=|\Gamma_1|^{n} h_{_{{\rm III}[3]}},
\end{eqnarray}
\end{subequations}
which follow from $\psi =\Gamma^m_{2}\phi = \Gamma^n_{1} \omega$.

$\phi$ satisfies
\begin{eqnarray}
  (a-c)  \dt \phi(l)&=&\phi(l)-\dt\phi(l+1),\\
  (b+c)  \h \phi(l)&=&\phi(l)+\h\phi(l+1),\\
  (a+b)\,\dt{\h \phi}(l)&=&\h \phi(l)+\dt\phi(l).
\end{eqnarray}
Similar to the previous case of $\psi$, using these relations we can get
\begin{subequations}
\label{II-bar}
\begin{eqnarray}
  (b+c)^{N-2} \h f_{_{{\rm II}[3]}} &=& |\h{N-2},\h \phi(N-2)|_{_{{\rm II}[3]}},\\
  (b+c)^{N-1} \h f_{_{{\rm II}[3]}} &=& |\h{N-2},\h \phi(N-1)|_{_{{\rm II}[3]}},\\
  -(a+b)(a-c)^{N-2}(b+c)^{N-2} \dt{\h f}_{_{{\rm II}[3]}} &=& |\h{N-3},\h \phi(N-2),\dt \phi(N-2)|_{_{{\rm II}[3]}};
\end{eqnarray}
\begin{eqnarray}
  (b+c)^{N-2} \bh f_{_{{\rm II}[3]}} &=& |\t{N-1},\h \phi(N-1)|_{_{{\rm II}[3]}},\\
  -(a+b) (a-c)^{N-2}(b+c)^{N-2} \dt{\bh f}_{_{{\rm II}[3]}}&=& |\t{N-2},\h \phi(N-1),\dt \phi(N-1)|_{_{{\rm II}[3]}},\\
  (b+c)^{N-1} \db{\h f}_{_{{\rm II}[3]}} &=& |-1,\h{N-3},\h \phi(N-2)|_{_{{\rm II}[3]}};
\end{eqnarray}
\begin{eqnarray}
  (b+c)^{N-2} [\h g_{_{{\rm II}[3]}} -(b+c) \h f_{_{{\rm II}[3]}}] &=& |\h{N-3},N-1,\h \phi(N-2)|_{_{{\rm II}[3]}},\\
  (b+c)^{N-2} [\h h_{_{{\rm II}[3]}}-(b+c) \h g_{_{{\rm II}[3]}}] &=& |\h{N-3},N,\h \phi(N-2)|_{_{{\rm II}[3]}}.
\end{eqnarray}
\end{subequations}
Now, noting that $f_{_{[3]}}= |\Gamma_2|^{m} f_{_{{\rm II}[3]}}$ and $\psi =\Gamma^m_{2}\phi$,
we have
\begin{equation}
(b+c)^{N-2} \h f_{_{[3]}} =(b+c)^{N-2}|\Gamma_2|^{m+1} \h f_{_{{\rm II}[3]}}
= |\Gamma_2||\psi(0),\cdots,\psi(N-2), \c E^{2} \psi(N-2)|_{_{[3]}},
\label{f-hat}
\end{equation}
where the operator $\c E^{2}$ is defined as \eqref{E-circ}, The above
formula is just \eqref{I-c} for $\mu=2$. From the rest of
\eqref{II-bar} we can get other formulas with an up-hat shift in
\eqref{Formula-I}-\eqref{Formula-I-h}, where for \eqref{I-l} and
\eqref{I-n} ($\mu=2,\nu=1$) we need to use \eqref{rela-I-b}.  Besides,
if we take a down-hat shift for \eqref{f-hat} and rewrite the last
column by using \eqref{rela-I-b} (with $\mu=\nu=2$), then we get
$$
2b (b+c)^{N-2}(b-c)^{N-2}f_{_{[3]}}
 = |\Gamma_{2}| |\h{N-3},\dh \psi(N-2),\c E^{2}  \psi(N-2)|_{_{[\kappa]}},
$$
i.e., the formula \eqref{I-l} of the case $\mu=\nu=2$.

Next, using $\omega$ we can derive these formulas with an up-tilde
shift in \eqref{Formula-I}-\eqref{Formula-I-h}.  $\omega$ satisfies
\begin{eqnarray}
  (a+c)  \t \omega(l)&=&\omega(l)+\t\omega(l+1),\\
  (b-c)  \dh \omega(l)&=&\omega(l)-\dh\omega(l+1),\\
  (a+b)\,\dh{\t \omega}(l)&=&\dh \omega(l)+\t\omega(l),
\end{eqnarray}
and from these relations we have
\begin{subequations}
\begin{eqnarray}
  (a+c)^{N-2} \t f_{_{{\rm III}[3]}} &=& |\h{N-2},\t \omega(N-2)|_{_{{\rm III}[3]}},\\
  (a+c)^{N-1} \t f_{_{{\rm III}[3]}} &=& |\h{N-2},\t \omega(N-1)|_{_{{\rm III}[3]}},\\
  (a+b)(a+c)^{N-2}(b-c)^{N-2} \dh{\t f}_{_{{\rm III}[3]}}
  &=& |\h{N-3},\dh \omega(N-2),\t \omega(N-2)|_{_{{\rm III}[3]}};
\end{eqnarray}
\begin{eqnarray}
  (a+c)^{N-2} \bt f_{_{{\rm III}[3]}} &=& |\t{N-1},\t \omega(N-1)|_{_{{\rm III}[3]}},\\
  (a+b) (b-c)^{N-2}(a+c)^{N-2} \dh{\bt f}_{_{{\rm III}[3]}}
  &=& |\t{N-2},\dh \omega(N-1),\t \omega(N-1)|_{_{{\rm III}[3]}},\\
  (a+c)^{N-1} \db{\t f}_{_{{\rm III}[3]}} &=& |-1,\h{N-3},\t \omega(N-2)|_{_{{\rm III}[3]}};
\end{eqnarray}
\begin{eqnarray}
  (a+c)^{N-2} [\t g_{_{{\rm III}[3]}} -(a+c) \t f_{_{{\rm III}[3]}}]
  &=& |\h{N-3},N-1,\t \omega(N-2)|_{_{{\rm III}[3]}},\\
  (a+c)^{N-2} [\t h_{_{{\rm III}[3]}}-(a+c) \t g_{_{{\rm III}[3]}}]
  &=& |\h{N-3},N,\t \omega(N-2)|_{_{{\rm III}[3]}}.
\end{eqnarray}
\end{subequations}
By using operator $\c E^{1}$, these formulas will generate those of \eqref{Formula-I}-\eqref{Formula-I-h}
with an up-tilde shift.

Thus we can get all formulas for $f$, $g$ and $h$ given by \eqref{Formula-I}-\eqref{Formula-I-h}.

\section{Formulas for $f=|\h{N-1}|_{_{[3]}},\quad  g=|-1,\t{N-1}|_{_{[3]}}$}
\label{A:c}

For the Casoratian $f_{_{[\kappa]}}$, we have
\begin{subequations}
\label{Formula-II}
\begin{eqnarray}
  E_{\mu} f
  \!\!&=& \!\! E_3 f- (\alpha_\mu-\alpha_3)|E_{\mu}\psi(-1), \h{N-2} |_{_{[3]}},\label{II-a}\\
  E_{\mu}E^3 f
  \!\!&=& \!\! f- (\alpha_\mu-\alpha_3)g +(\alpha_\mu-\alpha_3)^2|E_{\mu}\psi(-1), \t{N-1} |_{_{[3]}},\label{II-b}\\
  E_{\mu} g
  \!\!&=& \!\! |E_{\mu}\psi(-1), \h{N-2} |_{_{[3]}}- (\alpha_\mu-\alpha_3)|E_{\mu}\psi(-1),-1, \t{N-2} |_{_{[3]}},
  \label{II-c}\\
  \frac{(-1)^N}{|\Gamma_\mu|}E^{\mu} f
  \!\!&=& \!\! E_3 f- (\alpha_\mu+\alpha_3)|\c E^{\mu}\psi(-1), \h{N-2}|_{_{[3]}},\label{II-d}\\
  \frac{(-1)^N}{|\Gamma_\mu|}E^{\mu}E^3 f
  \!\!&=& \!\! f+(\alpha_\mu+\alpha_3)g -(\alpha_\mu+\alpha_3)^2|\c E^{\mu}\psi(-1), \t{N-1}|_{_{[3]}},\label{II-e}\\
  \frac{(-1)^N}{|\Gamma_\mu|}E^{\mu} g
  \!\!&=& \!\! -|\c E^{\mu}\psi(-1), \h{N-2} |_{_{[3]}}- (\alpha_\mu+\alpha_3)|\c E^{\mu}\psi(-1),-1, \t{N-2}|_{_{[3]}};
  \label{II-f}
\end{eqnarray}
\begin{align}
   &(a-b)\dth {\b f}\nonumber\\
  =&(a-b)\db f
  -(a-c)^2|\dt\psi(-1), \h{N-2} |_{_{[3]}}+(b-c)^2|\dh\psi(-1), \h{N-2} |_{_{[3]}} \nonumber \\
  & +(a-c)^2(b-c)|\dt\psi(-1), -1, \t{N-2}|_{_{[3]}}-(a-c)(b-c)^2|\dh\psi(-1),-1, \t{N-2}|_{_{[3]}}\nonumber\\
  & +(a-c)^2(b-c)^2|\dh\psi(-1), \dt\psi(-1), \t{N-2}|_{_{[3]}},\label{II-g}\\
  &\frac{a-b}{|\Gamma_1||\Gamma_2|}\th {\b f}\nonumber\\
  =&(a-b)\db f
  -(a+c)^2|\c E^{1} \psi(-1), \h{N-2} |_{_{[3]}}+(b+c)^2|\c E^{2}\psi(-1), \h{N-2} |_{_{[3]}} \nonumber \\
  & -(a+c)^2(b+c)|\c E^{1}\psi(-1), -1, \t{N-2}|_{_{[3]}}+(a+c)(b+c)^2|\c E^{2}\psi(-1),-1, \t{N-2}|_{_{[3]}}\nonumber\\
  & -(a+c)^2(b+c)^2|\c E^{2}\psi(-1), \c E^{1}\psi(-1), \t{N-2}|_{_{[3]}}.\label{II-h}
\end{align}
\end{subequations}

\section{Proof for the formulas in Appendix \ref{A:c}}

Noting that
\begin{equation}
\dt \psi(l+1)=\psi(l)-(a-c)\dt\psi(l),
\label{ff-1}
\end{equation}
we have
\begin{align*}
\dt f =&|\dt\psi(0),\dt\psi(1),\cdots,\dt\psi(N-3),\dt\psi(N-2),\dt\psi(N-1)|_{_{[3]}}\\
      =&|\dt\psi(0),\dt\psi(1),\cdots,\dt\psi(N-3),\dt\psi(N-2),\psi(N-2)-(a-c)\dt\psi(N-2)|_{_{[3]}}\\
      =&|\dt\psi(0),\dt\psi(1),\cdots,\dt\psi(N-3),\dt\psi(N-2),\psi(N-2)|_{_{[3]}}\\
      =&|\dt\psi(0),\dt\psi(1),\cdots,\dt\psi(N-3),\psi(N-3),\psi(N-2)|_{_{[3]}}\\
      =& \cdots\cdots\\
      =&|\dt\psi(0),\psi(0),\cdots,\psi(N-4),\psi(N-3),\psi(N-2)|_{_{[3]}}\\
      =&|\psi(-1)-(a-c)\dt\psi(-1),\psi(0),\cdots,\psi(N-4),\psi(N-3),\psi(N-2)|_{_{[3]}}\\
      =&\db f -(a-c)|\dt\psi(-1),\h{N-2}|_{_{[3]}}.
\end{align*}
This is \eqref{II-a} for $\mu=1$. In a similar way we can prove \eqref{II-b},
\eqref{II-c} and \eqref{II-g}.

To prove \eqref{II-d}, \eqref{II-e},
\eqref{II-f} and \eqref{II-h}, we consider the Casoratian type determinants
$f=|\h{N-1}|_{_{[3]}},\quad  g=|-1,\t{N-1}|_{_{[3]}}$ with the following entries
\begin{equation}
\mathbf{IV}: \sigma_i(n,m,l)=
  \varrho_{i}^{+}(c+ k_i)^l(a-k_i)^{-n}(b -k_i )^{-m}+
  \varrho_i^{-}(c- k_i)^l(a +k_i)^{-n}(b +k_i)^{-m},
\label{entry-H3-4}
\end{equation}
and denote  such $f$ and $g$ by $f_{_{{\rm IV}[3]}}$ and $g_{_{{\rm IV}[3]}}$.
Noting that
$\sigma$ satisfies
\begin{eqnarray}
  (a+c)  \t \sigma(l)&=&\sigma(l)+\t\sigma(l+1),\\
  (b+c)  \h \sigma(l)&=&\sigma(l)+\h\sigma(l+1),\\
  (a-b)\,\th \sigma(l)&=&\h \sigma(l)-\t\sigma(l),
\end{eqnarray}
we can have
\begin{subequations}
\label{IV-bar}
\begin{eqnarray}
  (-1)^N E^{\mu} f_{_{{\rm IV}[3]}}
  \!\!&=& \!\! E_3 f_{_{{\rm IV}[3]}}- (\alpha_\mu+\alpha_3)|E^{\mu}\sigma(-1), \h{N-2} |_{_{{\rm IV}[3]}},\\
  (-1)^N E^{\mu}E^3 f_{_{{\rm IV}[3]}}
  \!\!&=& \!\! f_{_{{\rm IV}[3]}}+ (\alpha_\mu+\alpha_3)g_{_{{\rm IV}[3]}}
  -(\alpha_\mu+\alpha_3)^2|E^{\mu}\sigma(-1), \t{N-1} |_{_{{\rm IV}[3]}},~~\\
  (-1)^N E^{\mu} g_{_{{\rm IV}[3]}}
  \!\!&=& \!\! -|E^{\mu}\sigma(-1), \h{N-2} |_{_{{\rm IV}[3]}}\nonumber\\
  && - (\alpha_\mu+\alpha_3)|E^{\mu}\sigma(-1),-1, \t{N-2} |_{_{{\rm IV}[3]}};
\end{eqnarray}
\begin{align}
   &(a-b)\th {\b f}_{_{{\rm IV}[3]}}\nonumber\\
  =&(a-b)\db f_{_{{\rm IV}[3]}}
  -(a+c)^2|\t\sigma(-1), \h{N-2} |_{_{{\rm IV}[3]}}+(b+c)^2|\h\sigma(-1), \h{N-2} |_{_{{\rm IV}[3]}} \nonumber \\
  & -(a+c)^2(b+c)|\t\sigma(-1), -1, \t{N-2}|_{_{{\rm IV}[3]}}
  +(a+c)(b+c)^2|\h\sigma(-1),-1, \t{N-2}|_{_{{\rm IV}[3]}}\nonumber\\
  & -(a+c)^2(b+c)^2|\h\sigma(-1), \t\sigma(-1), \t{N-2}|_{_{{\rm IV}[3]}}.
\end{align}
\end{subequations}
Then, using the relationship  $f_{_{{\rm I}[3]}}=|\Gamma_1|^{n} |\Gamma_2|^{m} f_{_{{\rm IV}[3]}}$,
$g_{_{{\rm I}[3]}}=|\Gamma_1|^{n} |\Gamma_2|^{m} g_{_{{\rm IV}[3]}}$ and $\psi =\Gamma_1^{n}\Gamma^m_{2}\sigma$,
we can finally derive the formulas \eqref{II-d}, \eqref{II-e}, \eqref{II-f} and \eqref{II-h}.

\end{appendix}
\end{document}